\documentclass[journal]{IEEEtran}

\usepackage{setspace}
\usepackage{lettrine}
\usepackage{graphics,
	psfrag,
	epsfig,
	amsthm,
	cite,
	amssymb,
	url,
	dsfont,
	balance,
	enumerate,
	color,
	setspace
}
\usepackage{amsmath}%[centertags][tbtags][sumlimits][nosumlimits][intlimits][namelimits][nonamelimits]

\usepackage{epstopdf}

\newtheorem{lemma}{Lemma}

\newtheorem{theorem}{Theorem}
\newtheorem{corollary}{Corollary}

\newcommand{\eref}[1]{(\ref{#1})}
\newcommand{\sref}[1]{Section~\ref{#1}}

\newcommand{\cref}[1]{Constraint~\ref{#1}}

\newcommand{\ignore}[1]{}

\usepackage[labelformat=simple]{subcaption}

\IEEEoverridecommandlockouts
\newcommand{\nosemic}{\renewcommand{\@endalgocfline}{\relax}}% Drop semi-colon ;
\newcommand{\dosemic}{\renewcommand{\@endalgocfline}{\algocf@endline}}% Reinstate semi-colon ;
\let\oldnl\nl% Store \nl in \oldnl
\newcommand{\nonl}{\renewcommand{\nl}{\let\nl\oldnl}}% Remove line number for one line

\usepackage{mathtools}
\usepackage{color}
\usepackage[utf8]{inputenc}
\usepackage[ruled, linesnumbered]{algorithm2e}
\SetEndCharOfAlgoLine{}

\begin{document}
	\IEEEoverridecommandlockouts
	%\begin{spacing}{1.44}
	\title{\vspace{-.9cm}  %Federated Learning via D2D Communications and Overlapped Clustering for Energy-Efficient Communication Systems
	Decentralized Aggregation for Energy-Efficient Federated Learning via Overlapped Clustering and D2D Communications}

	\setlength{\columnsep}{0.21 in}

	\author{
		\IEEEauthorblockN{Mohammed S. Al-Abiad, \textit{Member, IEEE}, Mohanad Obeed, Md. Jahangir Hossain, \textit{Senior Member, IEEE}, and Anas Chaaban, \textit{Senior Member, IEEE}} 
		
		\thanks {
			The authors are with the School of
			Engineering, the University of British Columbia, Kelowna, BC V1V 1V7, Canada
			(e-mail: m.saif@alumni.ubc.ca, mohanad.obeed@ubc.ca, jahangir.hossain@ubc.ca, anas.chaaban@ubc.ca).

		}
		%\vspace{-0.9cm}
	}
	
	\maketitle
	
	\begin{abstract}
	
Federated learning (FL) has emerged as a distributed
machine learning (ML) technique to train models without sharing users' private data. In this paper, we propose a decentralized FL
scheme that is called \underline{f}ederated \underline{l}earning \underline{e}mpowered  \underline{o}verlapped \underline{c}lustering for \underline{d}ecentralized aggregation (FL-EOCD). The introduced FL-EOCD leverages device-to-device (D2D) communications and overlapped clustering to enable decentralized aggregation, where a cluster is defined as a coverage zone of a typical device. The devices located on the overlapped clusters are called bridge devices (BDs). In the proposed FL-EOCD scheme, a clustering topology is envisioned where clusters are connected through BDs, so as the aggregated models of each cluster is disseminated to the other clusters in a decentralized manner without the need for a global aggregator or an additional hop of transmission.
Unlike the star-based FL, the proposed FL-EOCD scheme involves a large number of local devices by reusing the RRBs in different non-adjacent clusters. To evaluate our proposed FL-EOCD scheme as opposed to baseline FL schemes, we consider minimizing the overall energy-consumption of devices while maintaining the convergence rate of FL subject to its time constraint. To this end, a joint optimization problem, considering scheduling the local devices/BDs to the CHs and computation frequency allocation, is formulated, where an iterative solution to this joint problem is devised. Extensive simulations are conducted to verify the effectiveness of the proposed FL-EOCD algorithm over FL conventional schemes in terms of energy consumption, latency, and convergence rate.

\end{abstract}

	%\vspace{-0.55cm}
\section{Introduction}
%	\vspace{-0.1cm}
Machine learning (ML) techniques have been leveraged to provide various services, such as virtual reality, traffic prediction, and object recognition \cite{ML1, ML2, ML3}. The ML model training is usually implemented at the base station (BS) in a centralized manner, where the raw data is required to be collected at the BS for training coexist. However, in many applications, datasets are generated at edge devices, and consequently, transmitting raw data from all devices to BS incurs huge network traffic and devices data privacy leakage \cite{ML4}. Therefore, the traditional centralized ML techniques are not preferable for next generation  wireless networks \cite{Distributed_learning}. Federated learning (FL) has emerged as a promising distributed machine learning (ML) algorithms  for addressing the aforementioned bandwidth and privacy challenges \cite{FL2, FL3}. 
%FL is a distributed learning approach, where several models are trained in a distributive and collaborative manner at devices using their local data and sent to a FL server to aggregate the models.

\subsection{Federated Learning at the Network Edge}
	
The popular FL conventional architectures are: (i) star-based FL architecture \cite{C1, C2, C3, C4, C5, Ansari1} and (ii) relay-assisted/dual-hop FL \cite{Dual1, Dual2, Dual3, Dual4}, also called hierarchical FL. In each iteration of the star-based FL, each device trains a local model based on its own dataset and then transmits its local trained model to the BS. The BS then aggregates all the received models into one global model and broadcasts it back to the devices. These procedures are repeated several rounds until the global model converge. Similar to the star-based FL, the hierarchical FL requires the BS to aggregate all the local models of the devices, which are collected by  cluster-heads (CHs) (or relays). In particular, two-hop transmission is required in the hierarchical FL, i.e., the CHs/relays need to forward the local models of devices to the BS. Despite the fact that hierarchical FL can improve the limited connectivity of devices and reduce their energy consumption, its main limitation, as previously mentioned, is that it requires two-hop transmission. However, these FL conventional architectures pose prime concerns for the wireless edge since the devices may exhibit: (i) significant heterogeneity
in their computational resources (e.g., limited battery and CPU computations) \cite{fog}; (ii) varying proximity to the BS (e.g., varying distances from devices to the BS in a cell). As a consequence, a significant amount of energy is consumed to bounce the model between the BS and the devices \cite{C1}. This consumes the power of the battery-powered devices and the available resources at the BS that may be used for another set of devices. %This would make  battery-powered devices located at the network edge unsatisfied to transmit their local models to the distant BS. Moreover, such centralized approaches consume a lot of BS's resources and limit its ability to provide computational-intensive and time-sensitive mobile services to a different set of devices. %Despite the fact that hierarchical FL can improve the limited connectivity of devices and reduce their energy consumption, its main limitation, as previously mentioned, is that the models need to be transmitted over two hops. This leads to an increase in the communication cost and the delay. 
%With the recent progress of FL, new FL architectures were migrated from the aforementioned topology of conventional FL to mitigate the cost of uplink and downlink transmissions, e.g., \cite{fog, new1, new2}. However, these FL architectures still require a BS for global aggregations.

Towards energy-efficient FL while offloading the BS for performing global aggregations, we consider a new angle for developing resource-efficient and decentralized FL using device-to-device (D2D) communications \cite{D2D1}. D2D communications has been demonstrated to be a promising technology in 5G and beyond, such as fog computing and IoT systems \cite{D2D2, D2D3, D2D4, D2Dnew, D1, D2}; indeed, it is expected that $50\%$ of all network connections will be machine-to-machine by 2023 \cite{D2D5}. For example, D2D communications has been shown to be promising for significantly improving the energy efficiency, the reliability, and the data rate
\cite{D2D2, D2D3, D2D4, D2Dnew, D1, D2}. Unlike the hierarchical FL that can also use D2D clustering, we propose to overlap the clusters such that there is no need to have an additional level of aggregation at the BS, which may consume more energy. Therefore, it is imperative to exploit D2D communications and introduce a novel decentralized FL scheme for tackling the aforementioned issues of the FL conventional schemes.

Recently, a hybrid FL scheme that integrates the star-based topology with D2D communications, called semi-decentralized FL, was designed in \cite{semi}. The authors exploited D2D communications to propose  a consensus mechanism to mitigate model divergence via low-power communications among nearby devices. However, the main limitation of the semi-decentralized FL in \cite{semi} is that it uses a global aggregator and two-time scale FL (i.e., D2D communications and device-to-BS wireless transmissions). Different from the work in \cite{semi}, we design a novel decentralized FL scheme  that clusters  devices in a way that the models of each cluster is implicitly shared to the other clusters without the need for additional stage of aggregation at the BS.  To this end, we propose a novel decentralized FL and resource allocation D2D paradigm to minimize the overall energy consumption of a delay-constrained FL in a partially connected D2D system. In the envisioned system, each single-antenna device has a limited coverage zone, named cluster, and dynamically works as: (i) a local device; (ii) a bridge device (BD); or (iii) a CH that can aggregate the models of the associated local devices and BD in its cluster. %Each CH is granted a number of orthogonal radio resource blocks (RRBs), e.g., a group of OFDMA sub-carriers \cite{RRB1, RRB2}. Through these granted RRBs, the local devices upload their trained local models to the scheduled CHs for local cluster aggregations. %Besides their roles for performing local learning on their dataset, the BDs ensure the connectivity among CHs to exchange the received aggregated models between adjacent CHs, which will be used for the next round of iterations. Therefore, through the BDs, the local aggregated model of each CH propagates to other clusters without any additional communication cost or more consumed energy. 

%\textcolor{blue}{In this work, we propose a novel and decentralized FL and resource allocation D2D paradigm to minimize the overall energy consumption of a delay-constrained FL in the partially connected D2D-aided BS (PC-D2D-aided BS) system. In the envisioned system, each single-antenna device has limited coverage zone, named cluster, and works as a local learning device, a bridge device (BD), or a CH that is equipped with number of orthogonal radio resource blocks (RRBs), e.g., a group of OFDMA sub-carrier. The local learning devices train models on their datasets and transmit their local parameters to the scheduled CHs over RRBs for global model aggregation. The BDs, on the other hand, exchange the updated aggregated models between CHs. Therefore, through the BDs, the local aggregated model of each CH propagates to other clusters without any communication cost or more consumed energy.} \textcolor{red}{Therefore, we propose to cluster the devices in a way that the models of each cluster is implicitly shared to the other clusters with the help of BDs. Unlike the star-based FL, we propose to cluster the devices to involve a large number of devices and get the devices closer to the CH so the consumed energy is low. Unlike the hierarchical FL, we propose to overlap the clusters so that there is no need to have an additional level of aggregation that consumes more energy and increases the learning time.}

	%\vspace{-0.4cm}

\subsection{Related Works and Challenges}
	%\vspace{-0.1cm}
Since the devices are battery-driven, for a sustainable operation of an FL framework in  latency-sensitive applications, it is crucial to reduce the energy consumption of the edge devices and the FL time that consists of computation and communication latency. Therefore, we will focus on works addressing resource
efficiency and latency-constrained FL,
which is the main focus of this paper. Specifically, star-based and hierarchical latency-constrained FL schemes will be reviewed in this section and evaluated in the numerical results section. For a comprehensive survey of the FL literature, we refer the reader to e.g., \cite{Survey1}.

Several works extensively considered star-based FL from different aspects, such as frequency computation, wireless communications, and radio resource management, e.g., \cite{C1, C2, C3, C4, C5, Ansari1}. Another set of
works investigated the impact of wireless link conditions,  number of collaborating edge devices, and computation frequency allocation on the performance of star-based FL \cite{C3, C4, C5, Ansari1}. However, these works mainly relied on the BS for global aggregations. To
reduce the demand for global aggregations, a set of works considered a hierarchical system with two-hop transmission, e.g., \cite{Dual1, Dual2, Dual3, Dual4}, \cite{New}. For example, the authors in \cite{New} developed a hierarchical FL model where edge servers perform partial global aggregations. As mentioned before, such a hierarchical FL model leads to increasing the energy consumption and delay.

For developing energy-efficient FL, several works considered minimizing communication and computation energy, e.g., \cite{EE_FL_1, EE_FL_2, EE_FL_3, FCR11}. In \cite{EE_FL_1}, the authors proposed an energy-efficient radio resource allocation for delay constrained FL, where they minimized  the communication energy and ignored the computation energy. In \cite{EE_FL_3}, radio resource allocation was developed to minimize both the communication and computation energy in an FL system subject to delay constraints. The authors in \cite{FCR11}  proposed joint transmit power and computation frequency allocation to reduce the overall energy consumption of an FL in a fog-aided internet of things (IoT) network. Recently, a joint optimization framework considering device scheduling, transmit power allocation, and device's computation frequency allocation was proposed to simultaneously minimize the total energy consumption and maximize the number of the scheduled devices \cite{EE_FL_NOMA}. However, an energy-efficient framework, while
jointly considering communication and computation energy, for facilitating FL for network edge devices in partially connected D2D networks was not investigated in the existing literature.

Different from the aforementioned works, our goal is to propose a scheme,  called \underline{f}ederated \underline{l}earning \underline{e}mpowered \underline{o}verlapped \underline{c}lustering for \underline{d}ecentralized aggregation (FL-EOCD), that: (i) involves a decentralized FL scheme without a global aggregation at the BS; (ii) incorporates both local model training  and CHs' aggregated models dissemination among the CHs through the BDs; and (iii) develops an innovative clustering method leading to a computationally efficient solution. However, the proposed FL-EOCD scheme exhibits the following two challenges.
\begin{enumerate}
\item \textbf{Challenge I:} Unlike the aforementioned FL conventional systems where the optimization is performed only over the device scheduling and computation capability of collaborating devices, the optimization in the considered D2D environment also considers the selection of the set of CHs to achieve the
best FL network performance. Therefore, the performance of FL (i.e., accuracy and convergence) jointly depends on the selection of CHs, scheduled devices, RRB resource allocation, and computation frequency allocation of the scheduled devices. %In particular, the channel impairments, such as fading and interference can affect the FL accuracy; indeed a set of scheduled devices who are  far away from the selected CHs can affect the convergence of FL. 
Therefore, to fully capitalize the advantage of D2D communications, we must carefully consider a joint optimization of the CHs selection, device and RRB scheduling, and computation frequency allocation of the devices.
\item \textbf{Challenge II:} In practice, there is a trade-off	between the FL time (i.e., computation time and wireless transmission latency) and edge device energy consumption; since reducing the FL time
requires more energy. However, the computation time of devices is affected by the allocated computing resources, and increasing the computation frequency of the devices can reduce the	latency at the cost of the increased energy consumption. Furthermore, the wireless data transmission latency is determined by the wireless transmission rate, 	which is related to wireless transmission power. To reduce the wireless communication latency, we can increase the transmission power, however this will consume more energy. Accordingly, a joint		optimization of the degrees-of-freedom, namely, CHs selection, device scheduling to the CHs/RRBs, and computation frequency allocation, is required
to satisfy a given FL time constraint and reduce energy consumption. %Therefore,	we propose
%a novel optimization scheme capturing the highly complex interplays between energy consumption, wireless communication latency, and  computation frequency allocation while satisfying the FL time requirements. %To this end, as we will show, our novel approach has promising results in terms of energy consumption and  FL time (consisting of computation and communication latencies)  as opposed to the conventional FL architectures. 
\end{enumerate}

%	\vspace{-0.4cm}
\subsection{Contributions}
%	\vspace{-0.1cm}
%In this paper, we introduce an FL-EOCD scheme by leveraging D2D communications and overlapped connectivity between the devices. To optimize the proposed scheme, we formulate and solve the problem of device-CH/RRB scheduling and computation frequency allocation in the partially connected D2D network. To this end, we develop an innovative clustering method leading to a computationally efficient solution. 
The specific contributions of this work are summarized as follows.
\begin{itemize}
\item We introduce a decentralized FL scheme, called FL-EOCD, which augments the conventional star-based and hierarchical FL architectures with overlapped clustering that enables  global models dissemination without the need for a central aggregation at the BS. In FL-EOCD, the local models are trained at the 	devices and transmitted to the scheduled CHs for aggregation. Meanwhile, the BDs exploit their connectivity between the CHs to disseminate the received aggregated models among all the CHs in the network. Different from conventional FL schemes, FL-EOCD scheme clusters devices in a way that the models of each cluster is implicitly shared to the other clusters, where the aggregation is implemented in a distributed manner. To the best of the authors’ knowledge, this is the
first work that introduces a decentralized FL scheme to efficiently deal with convergence rate, energy-efficiency, and learning time without involving a BS for aggregations.
	
\item 	We develop an energy-efficient framework to facilitate the introduced FL-EOCD scheme for the energy-limited devices in the partially connected D2D network. We formulate an  optimization	problem to minimize both the computation	and communication energy consumption of the devices	subject to constraints on FL time, device connectivity, scheduling	among the device-CH/RRB, and computation frequency allocation. The proposed joint optimization problem is NP-hard, and thus the global optimal solution
is computationally intractable. To tackle such difficult optimization
problem, an iterative solution framework is devised. Specifically, the overall CH selection, device-CH scheduling, and  computation frequency allocation problem is
decomposed into two sub-problems, namely,	CHs selection and their devices scheduling sub-problem and computation frequency allocation sub-problem. For each sub-problem, we propose a solution of polynomial computational
complexity, and rigorously justify the effectiveness
of the proposed solutions. By iteratively solving these sub-problems, an FL-EOCD algorithm is developed to minimize the energy consumption.

	\item The effectiveness of the proposed FL-EOCD scheme is verified and compared with the existing FL schemes via simulations. Simulation results show that the proposed FL-EOCD scheme achieves $55\%$ ($30\%$) and $35\%$ ($50\%$) lower energy consumption and shorter FL time, respectively, compared to the considered star-based and hierarchical FL schemes. In terms of convergence rate, if the required accuracy is low (e.g., $80\%$), the proposed scheme requires more global iterations (around $3\%$ more) than the star-based and the hierarchical FL schemes. If the required accuracy is high (e.g., $>93\%$), the proposed, star-based, and hierarchical schemes need the same number of global iterations.
\end{itemize}

The rest of this paper is organized as follows. 
The system model is described in \sref{SMMM}. The optimization problem formulation is provided in \sref{PF}. The proposed solution is presented in \sref{J}. \sref{A} presents an overview of the properties of the FL-EOCD algorithm. The
simulation results and the concluding remarks are provided in \sref{NR} and \sref{CN}, respectively.

%	\vspace{-0.55cm}

\section{System Model and Federated Learning } \label{SMMM}
%\vspace{-0.4cm}
\begin{figure}[t!]
	\centerline{\includegraphics[width=0.65\linewidth]{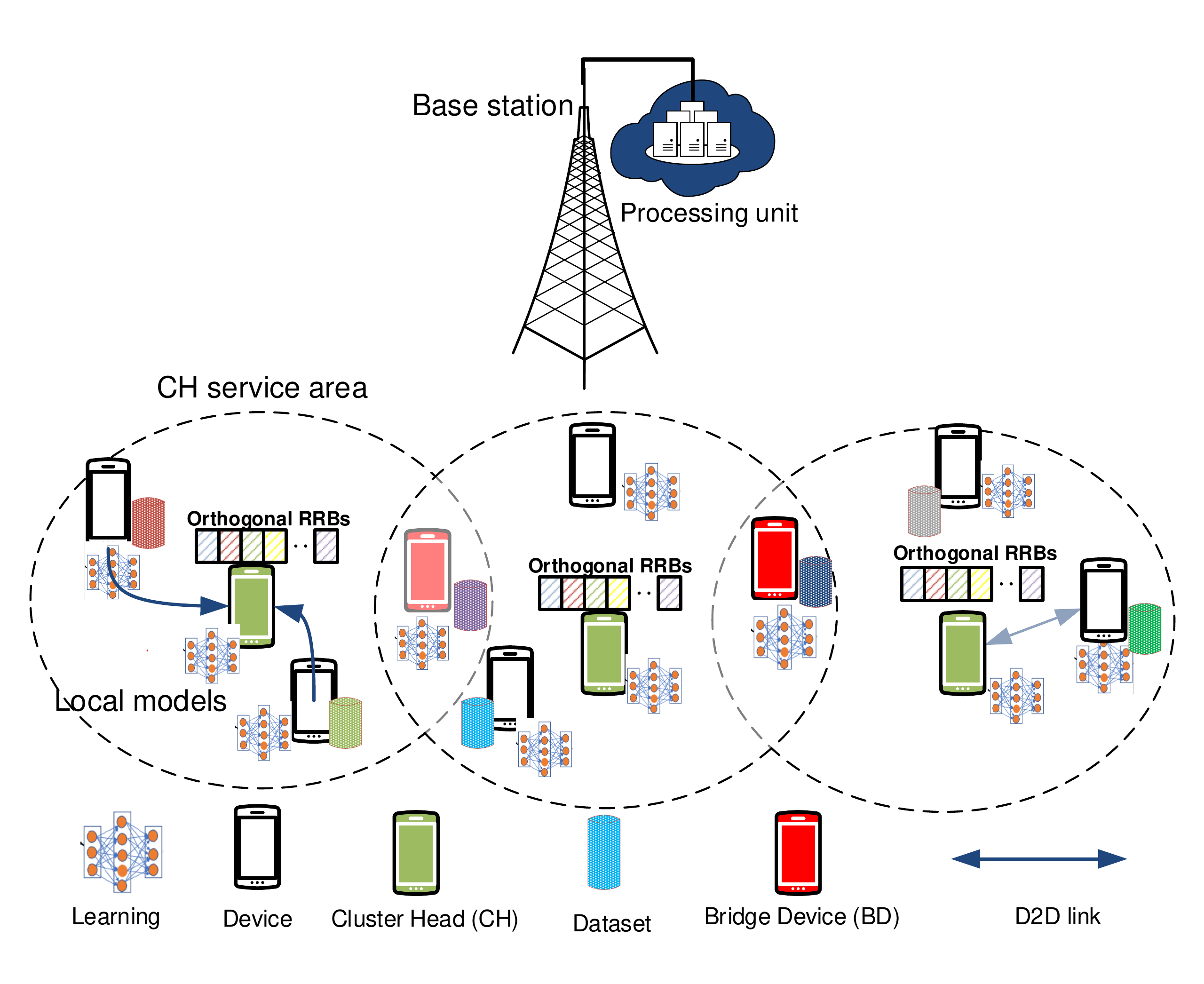}}
		\caption{Illustration of a partially connected D2D network.}
	\label{fig1}
	%\vspace{-1.5em}
\end{figure}
	
\subsection{System Overview}
	%\vspace{-0.25cm}
We consider a partially connected D2D network with a single BS and a set of devices that is denoted
by  $\mathcal{N}=\{1,2,\cdots,N\}$. An example of the envisioned system model is illustrated in Fig. \ref{fig1}. Each device $n$ has a limited coverage zone that represents the service area of that device within a circle of radius $\mathtt R$. Such a service area is denoted in this work as a cluster. It is assumed that the clusters are formed based on the ability of devices to conduct low-energy D2D communications, e.g., geographic proximity. Therefore, the set of devices in the $n$-th device's cluster is defined by $\mathcal A_{n}=\{m\in \mathcal{N}| d_{n,m}\leq \mathtt R$\}, where $d_{n,m}$ is the distance between the $n$-th device and the $m$-th device.  It is assumed that
each device cannot transmit and receive at the same time, i.e., half-duplex D2D channel is adopted, and each device is equipped with a single antenna.  Meanwhile, the BS in this work is not engaged in the FL global aggregation model and it is responsible for the scheduling process only. In particular, it decides the CH/BD selection and local device-CH association and delivers the scheduling parameters to the CHs/BDs over reliable control and non-payload communication (CNPC) channels.  %\ignore{If a part of the PC-D2D-aided BS network is disconnected,
%it can be considered as an independent network and optimized separately.In this work, we do not consider any constraints on the composition of devices within a cluster, as long as they
%follow a common D2D protocol []-[].}

The FL model training is considered through a sequence of
global iterations indexed by $t=\{1,2,\cdots, T\}$ as it will be explained in the next subsection. In each global iteration, devices have fixed locations and can change their locations in different global iterations \cite{semi, Ansari1}. Subsequently, at any FL global iteration, the set of devices $\mathcal N$ is divided into: i) a subset of devices that can perform local learning on their dataset, this subset of devices are called local devices, ii) a subset of devices that can perform local learning and aggregate the local models of devices in their clusters which are denoted by cluster heads (CHs), and iii) a subset of devices that can perform local learning and combine the received aggregated models from adjacent CHs with their own local models and forward them back to the adjacent CHs for aggregated model disseminations; this subset of devices are called bridge devices (BDs). Let $\mathcal{N}_l=\{1,2,\cdots,N_l\}$ be the subset of local devices; $\mathcal{C}=\{1,2,\cdots,C\}$ is the subset of CHs; $\mathcal{B}=\{1,2,\cdots,B\}$ is the subset of BDs; where $\mathcal N_l\cup \mathcal C\cup\mathcal B \in \mathcal N$. For simplicity and unless otherwise stated, $\mathcal N_l$ represents the set of scheduled devices and BDs, since the BDs can also perform local learning.

We propose that the local devices and BDs
transmit their local models to their assigned CHs via D2D links using orthogonal RRBs in an uplink phase. Specifically, each CH is granted a limited number of $Z$ orthogonal RRBs that are denoted by the set $\mathcal{Z}=\{1,2,\cdots,Z\}$, where local devices can use them to transmit their trained models to the scheduled CHs. 
Each RRB is used to denote a time/frequency resource block of each CH, i.e., a group of  orthogonal sub-carriers \cite{RRB1, RRB2}. In this work,  we consider that each local device is scheduled to only one RRB, and each RRB is assigned to only one local device. Thus, we consider the following two binary optimization variables for local device-CH assignment: (i) $s_{n,c}=1$ if the $n$-th device is assigned to the $c$-th CH and $s_{n,c}=0$ otherwise; and (ii) $r^n_{c,z}=1$ if the $n$-th device is allocated to the $c$-th CH on the $z$-th RRB, and $r^n_{c,z}=0$ otherwise. %Hence, $\mathbf S=\{s_{n,c}\}$ represents
%the scheduling matrix of devices $\mathcal N_l$ to the set of CHs $\mathcal N_c$ and $s_{n,c} \in \{0,1\}$ is the element of the $n$-th row and $c$-th column of $\mathbf S$. Similarly,  $\mathbf R=\{r^n_{c,z}\}$ represents the RRB allocation matrix and  $r^n_{c,z}=1$ is an element in  $\mathbf R$. 

Let $p_n$ denote the transmission power of the $n$-th device, and consequently, the achievable rate of the $n$-th device to the $c$-th CH over the $z$-th RRB can be given by $R^n_{c,z}=W\log_{2}(1+\frac{r^n_{c,z}p_n \left|h^n_{c,z}\right|^2}{N_0}), \forall c\in \mathcal A_n$, where $W$ is the bandwidth per RRB,  $N_0$ denotes the additive white Gaussian noise (AWGN) variance, and $h^n_{c,z}$ denotes the D2D channel  between the $n$-th device and the $c$-th CH over the $z$-th RRB. Similarly, let $h^n_{c}$ denote the D2D channel  between the $c$-th CH and the $n$-th scheduled device. Then, the achievable rate of D2D pair $(c, n)$ is given by $R^n_{c}=W\log_{2}(1+\frac{s_{n,c}p_c \left|h^n_{c}\right|^2}{N_0})$. Consequently, the $c$-th CH adopts a common transmission rate $R_c$ that is equal to the minimum achievable rates of all its scheduled devices that is denoted by $\mathcal N_l^c$. This adopted transmission rate is $R_c=\min_{n\in \mathcal N_l^c }R^n_{c}$.  

\subsection{Federated Learning Process}
%	\vspace{-0.33cm}
We propose an FL-EOCD scheme that has distinguishing features compared to the conventional FL schemes, e.g., star-based FL models \cite{C1, C2, C3, C4, C5, Ansari1} and hierarchical FL models \cite{Dual1, Dual2, Dual3, Dual4}. %, and decentralized FL models \cite{fog, new1, new2}.
  Particularly, our envisioned scheme learns the local models at devices and aggregates a shared global model through CHs and BDs without the need for global aggregations at the BS.

Let $\mathcal D_n$ denote the local data set of the $n$-th local device, which represents a set of data samples $\{x_i, y_i\}$, where $x_i$ is the sample $i$’s input (e.g., image
pixels) and $y_i$ is the sample $i$’s output (e.g., label of the image). Here, $\mathcal D_n=\{(x_i, y_i): i=1, \cdots, D_n\}$ denotes the dataset involved in the training process of the $n$-th local device. The total dataset involved in the training process in the system is $\mathcal D$. A widely used model in linear regression, logistic regression, and deep neural networks is considered, where the local loss function on the dataset of the $n$-th local device can be calculated as
\begin{equation}\label{L_n}
L_n(\mathbf w)=\frac{1}{|\mathcal D_n|}\sum_{(x_i,y_i)\in \mathcal D_n}l_i(\mathbf w), \forall n\in \mathcal N_l,
\end{equation} 
where $|\mathcal D_n|$ is the total number of data samples of the $n$-th local device and $\mathbf w$ is the learning model parameter vector. For simplicity,
we define $D_n=|\mathcal D_n|$.  Here, $l_i(\mathbf w)$ is the loss function that measures the local training model error of the $i$-th data sample. Common examples of loss function  $l_i(\mathbf w)$ include linear regression with $l_i(\mathbf w)=0.5\|\boldsymbol x^T_i\mathbf w-y_i\|^2$ and support vector machine with $l_i(\mathbf w)=\max\{0,1-y_i\boldsymbol x^T_i \mathbf w\}, y_i\in \{-1,1\}$ \cite{C4, LC}. In what follows, we explain the FL process of our advised learning model.

\textbf{(1) FL process:} In the envisioned FL-EOCD process, there are global iterations and local iterations. The specific process of the proposed FL-EOCD  at the $t$-th global iteration can be summarized as follows, which is presented in Fig. \ref{fig2_2}. 

\begin{figure}[t!]
	\centerline{\includegraphics[width=0.85\linewidth]{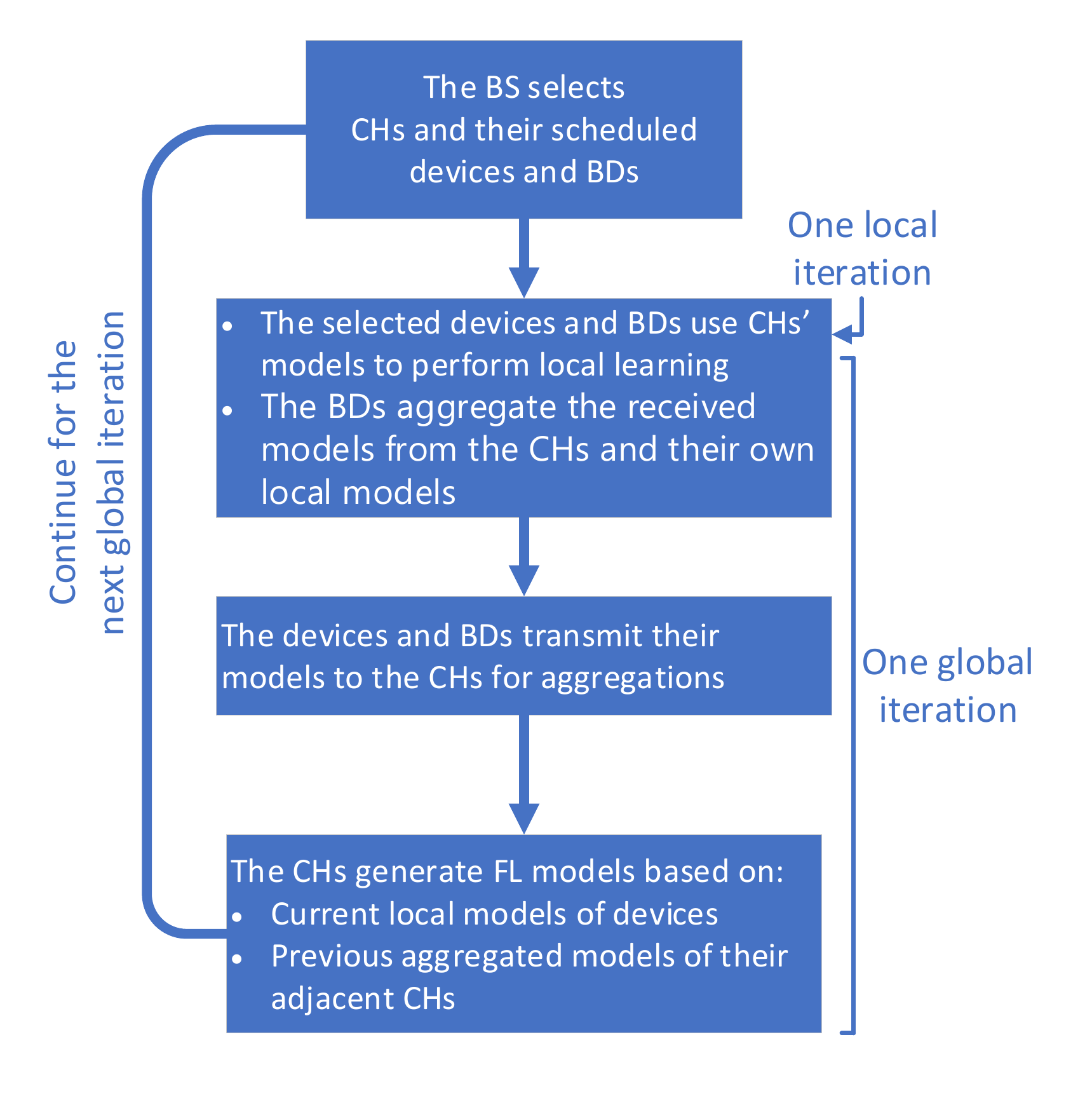}}
	\caption{The training procedure of the proposed FL.}
	\label{fig2_2}
	%\vspace{-1.5em}
\end{figure}

\begin{figure}[t!]
	\centerline{\includegraphics[width=0.85\linewidth]{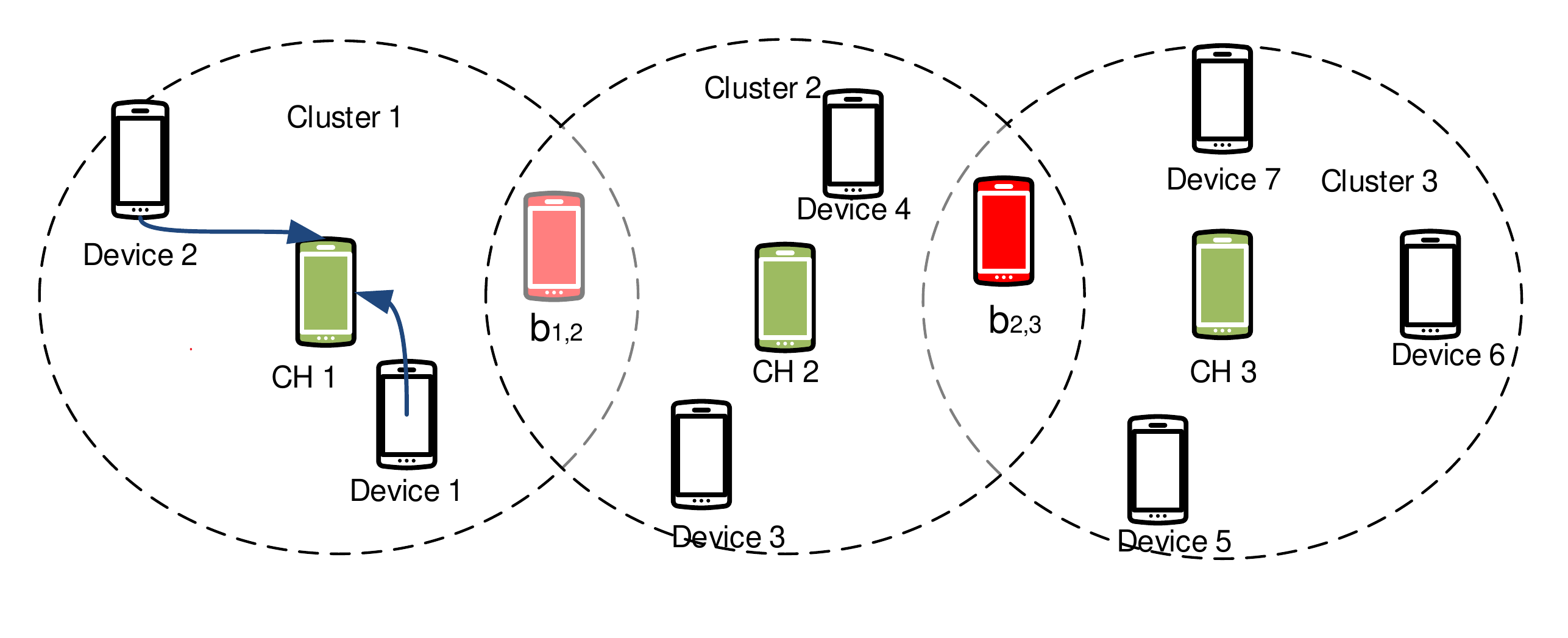}}
	\caption{FL model containing a set of $12$ devices that is divided into: (i) $C=3$ CHs for local aggregations, (ii) $B=2$ BDs for learning and for exchanging CHs' models between adjacent CHs, and (iii) $N_l=7$ local devices for learning.}
	\label{fig2}
	%\vspace{-1.5em}
\end{figure}

\ignore{\begin{figure}[t!]
	\centering
	\begin{minipage}{0.494\textwidth}
		\centering
		\includegraphics[width=0.65\textwidth]{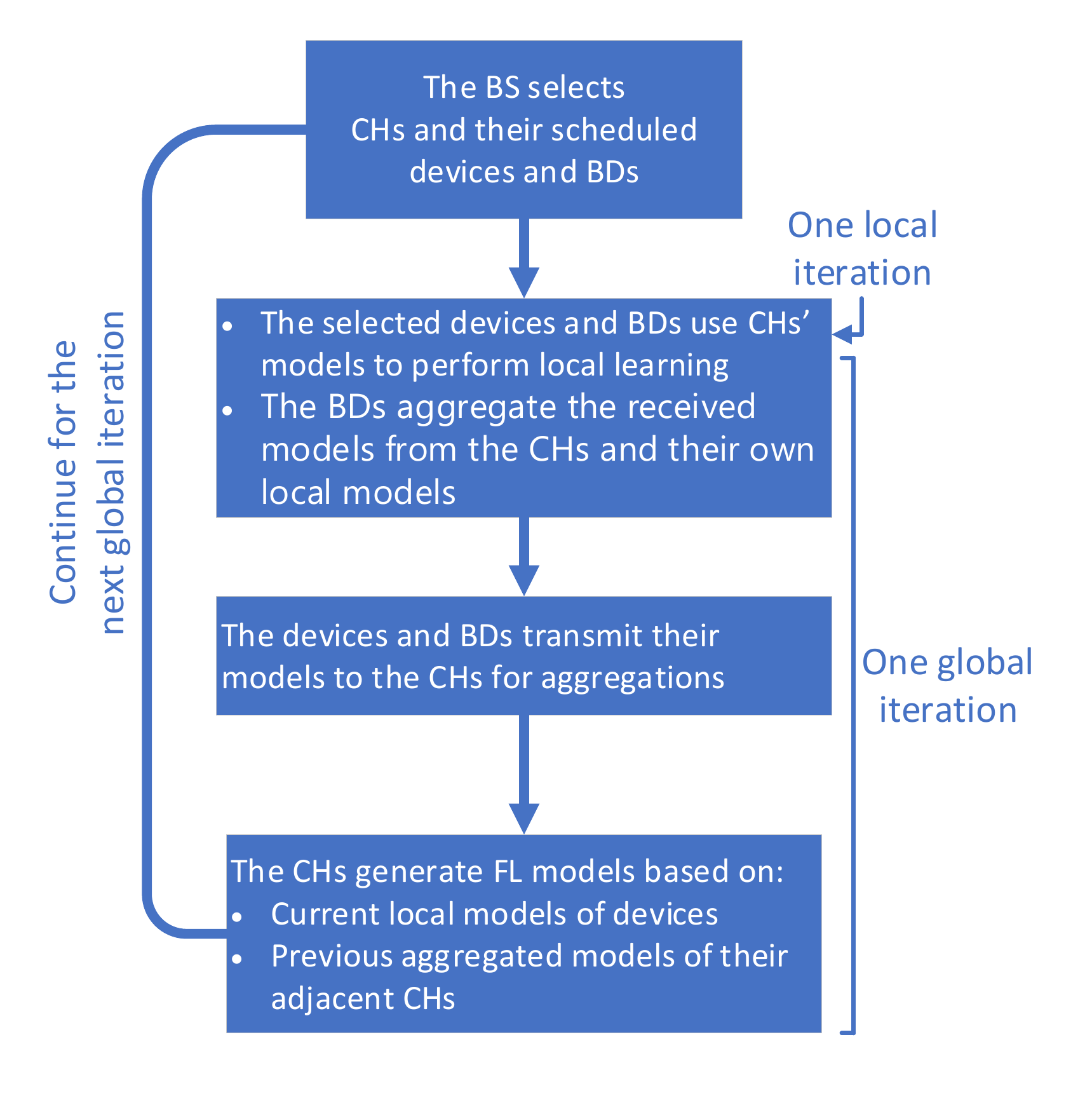} % first figure itself
		\caption{The training procedure of the proposed FL.		}
		
		\label{fig2_2}
	\end{minipage}\hfill
	\begin{minipage}{0.494\textwidth}
		\centering
		\includegraphics[width=0.65\textwidth]{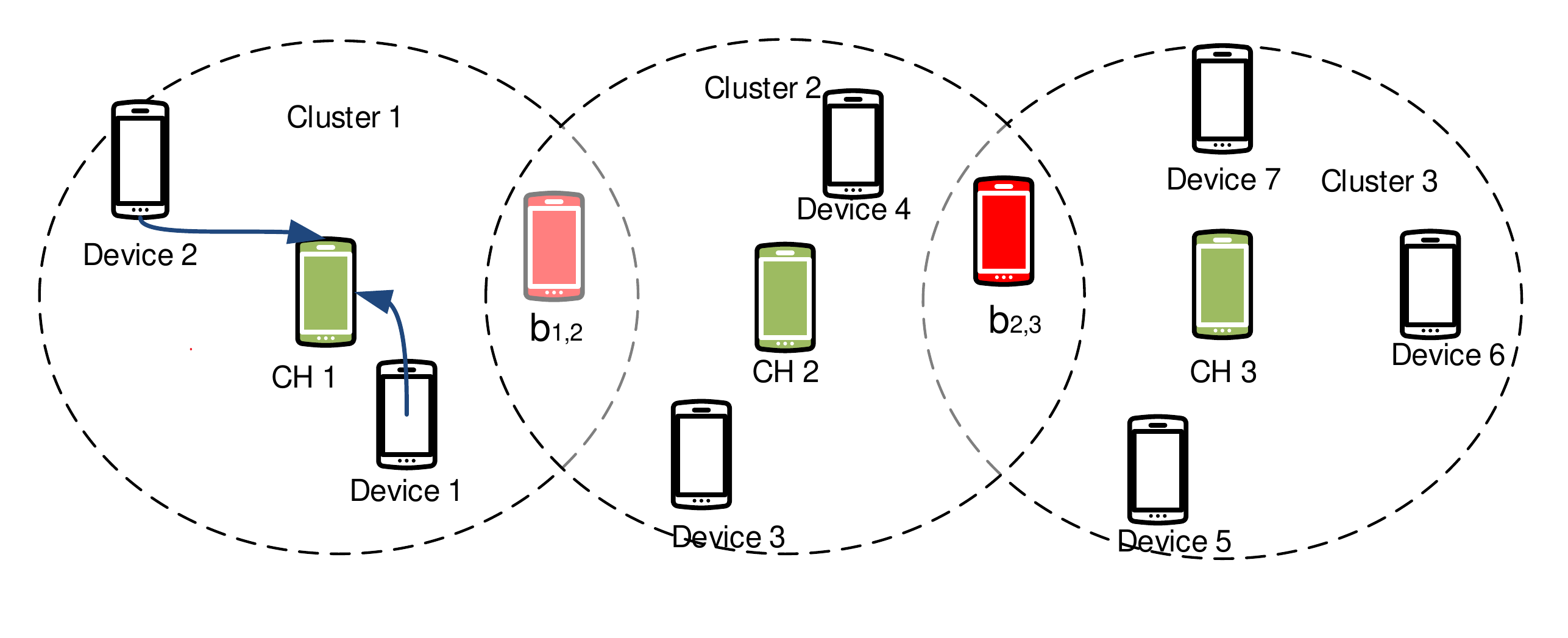} % first figure itself
		\caption{FL model containing a set of $12$ devices that is divided into: (i) $C=3$ CHs for local aggregations, (ii) $B=2$ BDs for learning and for exchanging CHs' models between adjacent CHs, and (iii) $N_l=7$ local devices for learning.}
		
		\label{fig2}
	\end{minipage}\hfill
\end{figure} 
}

\begin{itemize}
	\item The BS judiciously determines a subset of local devices $\mathcal N_l=\{1, 2, \cdots, N_l\}$ to independently train local ML models and  subsets of CHs and BDs;
	\item Each CH $c$ sends its updated aggregated model $\mathbf w_c(t-1)$ to its scheduled devices $\mathcal N^c_l$ and the associated BDs in a downlink phase. Note that adjacent CHs  transmit the aggregated models simultaneously using orthogonal RRBs and the CHs of the non-adjacent clusters can transmit their aggregated models to the selected local devices on the same RRB.%, while adjacent CHs use orthogonal RRBs. %Thanks to the realistic partially D2D network topology,	where CHs can only serve the subsets of devices in their coverage;
	
	\item Each selected device $n\in \mathcal N_l$ executes a local update algorithm 	(e.g., stochastic gradient descent (SGD) algorithm) based on its	dataset $\mathcal D_n$ and the aggregated model $\mathbf w_c(t-1)$ of its associated CH $c$. The output of this local learning at the $n$-th device is 	the updated local model $\mathbf w^*_{c,n}(t)$;

	\item Meanwhile, BDs receive the aggregated models $\mathbf w_c(t-1)$ from the associated CHs $\mathcal C$,  perform a local learning algorithm, aggregate the received models and their own updated models, forward the aggregated models to the corresponding CHs at the uplink phase of each iteration.  This ensures the dissemination of the aggregated models of the CHs over the entire network. On the other hand,	each CH can subtract its previous model from the received aggregated model of the BDs to retrieve the aggregated model of the adjacent CH(s) as will be explained next.

\end{itemize}

We define the loss function for the $c$-th CH as the average local loss across the $c$-th cluster,
\begin{equation}\label{L_c}
\hat{L}(\mathbf w)=\frac{1}{\sum_n|\mathcal D_n|}\sum^{N^c_l}_{n=1}|\mathcal D_n|L_n(\mathbf w).
\end{equation} 
The global loss function $L(\boldsymbol \omega)$ is then defined as the average loss across all the clusters,
\begin{equation}\label{L_t}
L(\mathbf w)=\frac{1}{\sum_n|\mathcal D_n|}\sum^{N_l}_{n=1}|\mathcal D_n|L_n(\mathbf w).
\end{equation} 
The objective of the FL model training is to find the optimal model
parameters $\mathbf w^*$ for $L(\mathbf w)$ that is expressed as follows $\mathbf w^*=\arg\min_{\mathbf w} L(\mathbf w)$. Next, we explain in details the operations that are performed at the CHs and the BDs.

\textbf{(2) BD and CH aggregation process:} For simplicity of the ensuing analysis, we explain the operations that should be implemented at the CHs and BDs through an example.  Consider an example of the envisioned FL model that is shown in Fig. \ref{fig2}, where the CHs and their scheduled devices and BDs are given. Suppose that each CH is granted $4$ RRBs where local devices can transmit their local models. We first explain the operations of the CHs and BDs based on the overlapped chain clustering as in Fig. \ref{fig2}, then we generalize the operations for any number of clusters and any form of clustering structure. For notations convenience, here, we use $c_i$ to denote the $c_i$-th CH, while in the remaining sections we use the notation $c$ for the $c$-th CH.  %\footnote{For simplicity of the ensuing analysis, this example considers a chain of clusters, where CHs in the middle of the chain are associated with a maximum of two BDs. However, adjacent CHs can be practically associated with multiple BDs. Accordingly, we consider the more general network setup, in which case the CH can be associated with multiple BDs.}.  
%For the FL process illustration, we next explain the operations occurred on the CHs and BDs over the global iterations.

\textbf{Iteration 1:} The $c_1$-th, $c_2$-th, and $c_3$-th CHs transmit their initial models $\mathbf w_{c_i}(0)$ to the associated devices. For example, the $c_1$-th CH transmits its model $\mathbf w_{c_1}(0)$ to the $n_1$-th and $n_2$-th devices, and the $b_{1,2}$-th BD. The $c_1$-th and $c_3$-th CHs  can broadcast their models on the same RRB, while the $c_2$-th CH uses different RRB. Accordingly, the $b_{1,2}$-th and $b_{2,3}$-th BDs receive the models from the $c_1$-th and $c_2$-th CHs and the $c_2$-th and $c_3$-th CHs, respectively. 
Afterword, each device executes the local learning algorithm, whose output is $\mathbf w^{*}_{c_i,n}(1)$, i.e.,  $\mathbf w^{*}_{1,1}(1)$ represents the local model of the $n_1$-th device that is scheduled to the $c_1$-th CH. While the devices do the local learning, the  $b_{i,j}$-th BD picks one model from the received ones (i.e., $\mathbf w_{c_i}(0)$ or $\mathbf w_{c_j}(0)$) or both of them after aggregation to update its own model $\mathbf w_{d_{ij}}(0)$ using its available dataset. Hence, the $d_{i,j}$-th BD has now three models which are $\mathbf w_{b_{i,j}}(1)$, $\mathbf w_{c_i}(0)$, and $\mathbf w_{c_j}(0)$. The $b_{i,j}$-th BD then aggregates these three models as follows 
\begin{multline}
\label{BD_agg}
    \mathbf w_{b_{i,j}}^{ag}(1)= \frac{1}{D_{c_i}+D_{c_j}+D_{b_{i,j}}}(D_{c_i}\mathbf w_{c_i}(0) + D_{c_j}\mathbf w_{c_j} (0)\\ + D_{b_{i,j}} \mathbf w_{b_{i,j}}(1)  ),
\end{multline}
where $D_{c_i}$ and $D_{c_j}$ are the total number of data points of the local devices that are associated to the $c_i$-th and $c_j$-th clusters, respectively, and $D_{b_{i,j}}$ is the number of data points at the $b_{i,j}$-th BD. At the uplink phase of the first iteration, each local device transmits its updated model, while each BD transmits the aggregated model given in \eqref{BD_agg}. 

\textbf{Iteration 2:}  Now, each CH has received different models from the associated local devices and BDs. We propose that each CH first aggregates the models of the local devices (without including the models of the BDs) as follows
\begin{equation}
\label{wcD}
    \mathbf w_{c_i}(1) = \frac{1}{D_{c_i}} \sum_{n=1}^{N_l^{c_i}} D_n \mathbf w_{c_i,n}^*(1).
\end{equation}
Then, the $c_i$-th CH extracts the models of the other clusters aggregated with the BD's model  from the models received by the BDs. For instance, the $c_2$-th CH obtains the model $D_{c_1}\mathbf w_{c_1} (0) + D_{b_{1,2}} \mathbf w_{b_{1,2}}(1)$ from the received $\mathbf w_{b_{1,2}}^{ag}(1)$ as follows
\begin{multline}
   D_{c_1}\mathbf w_{c_1} (0) + D_{b_{1,2}} \mathbf w_{b_{1,2}}(1) = (D_{c_1}+D_{c_2}+D_{b_{1,2}})\mathbf w_{b_{1,2}}^{ag}(1)\\ - D_{c_2}\mathbf w_{c_2}(0),
\end{multline}
where $\mathbf w_{c_2}(0)$ is known at the $c_2$-th CH from the previous iteration. Similarly, the  $c_2$-th CH can obtain the model $D_{c_3}\mathbf w_{c_3} (0) + D_{b_{2,3}} \mathbf w_{b_{2,3}}(1)$ from the received $\mathbf w_{b_{2,3}}^{ag}(1)$. In this way, each CH can extract the models of the adjacent cluster(s) that are aggregated with the corresponding BD model. Now each CH aggregates the model given by \eqref{wcD} with the extracted models from the BDs and broadcasts them back to the associated devices. For example, the transmitted model of the $c_2$-th CH is given by 
\begin{align}
\label{wcag}
    \mathbf w_{c_2}^{ag}(1) = &\nonumber \frac{1}{D_{c_2}+D_{c_1}+D_{b_{1,2}}+D_{c_3}+D_{b_{2,3}}}\bigg(D_{c_2}\mathbf w_{c_2}(1)\\& \nonumber +D_{c_1}\mathbf w_{c_1}(0) +D_{c_3}\mathbf w_{c_3}(0) 
    + D_{b_{1,2}} \mathbf w_{b_{1,2}}(1) \\& + D_{b_{2,3}} \mathbf w_{b_{2,3}}(1) \bigg).
\end{align}
Each device then can update its local model based on the available dataset and the received model from the associated CH. Whereas, each BD 
 first removes the redundant models that has been received from the previous round, conducts the local learning, sums the aggregated models of the adjacent CHs, and transmits back to the associated CHs. For example, the $b_{1,2}$-th BD first removes the terms that contain $\mathbf w_{c_1}(0)$ and $\mathbf w_{b_{1,2}}(1)$ from the received model $\mathbf w_{c_2}^{ag}(1)$ that is given by \eqref{wcag}, since these terms are known from the previous iteration. % by removing the terms that contain $\mathbf w_{c_1}(0)$ and $\mathbf w_{12}(1)$,
   Particularly, at the second iteration, for instance the $b_{1,2}$-th BD transmits the following to the $c_1$ and $c_2$-th CHs:
  
 \begin{align}
     \mathbf w_{b_{1,2}}^{ag}(2) =& \nonumber  \frac{1}{D_{c_2}+D_{c_1}+D_{b_{1,2}}+D_{c_3}+D_{b_{2,3}}}\bigg(D_{c_2}\mathbf w_{c_2}(1) \\& \nonumber +D_{c_1}\mathbf w_{c_1}(1) +D_{c_3}\mathbf w_{c_3}(0) + D_{b_{1,2}} \mathbf w_{b_{1,2}}(2)\\&  +D_{b_{2,3}} \mathbf w_{b_{2,3}}(1) \bigg ).
 \end{align}
 
 It is important to note that once the model received by a CH or a BD, they first remove the models that they know from the previous iterations.
 
 \textbf{Iteration 3 and beyond:}  At the third iteration, we can notice that the model of the $c_3$-th CH is attained by the $c_1$-th CH and vise versa but with a delay of two iterations. In the third iteration and in the next ones, each CH first aggregates the models of the associated devices, cleans the models received from the associated BDs by removing the known previous models, aggregates the devices' model with the models received by the associated BD(s), and then broadcasts back to the associated devices. Whereas, each BD first cleans the received models, conducts local learning to update its own model, aggregates the received models with the updated one, transmits to the associated CHs. 
 
  In general, at the $t$-th iteration (when $t\geq C$), the aggregated model at the $c_i$-th CH is given at the top of the next page, 
  \begin{table*}
\begin{equation}
    \label{t_iter}
    \mathbf w_{c_i}^{ag}(t-1)= \frac{1}{\sum_{j=1}^{C} D_{c_j}+\sum_{q=1}^{B} D_{b_{q,q+1}}}\bigg(\sum_{j=1}^{C} D_{c_j}\mathbf w_{c_j}(t-1-y_{i,j}) + \sum_{q=1}^{B}  D_{{q,q+1}}\mathbf w_{b_{q,q+1}}(t-1-v_{i,q}) \bigg),
\end{equation}
\hrulefill
\vspace*{-0.5cm}
\end{table*}
where $y_{i,j}$ is the lowest number of BDs that are between the $c_i$-th and  $c_j$-th clusters, and $v_{i,q}$ is the lowest number of clusters between the $c_i$-th CH and $b_{q,q+1}$-th BD. For example, in Fig. \ref{fig2}, there is a single BD between the  $c_1$-th and $c_2$-th CHs and two BDs between the $c_1$-th and $c_3$-th CHs, so $y_{1,2}=1$ and $y_{2,3}=2$. On the other hand, there is no cluster between the $c_1$-th CH and  $b_{1,2}$-th BD and a single cluster between the $c_1$-th CH and $b_{2,3}$-th BD, so $v_{1,1}=0$ and $v_{1,2}=1$. Note that  $y_{i,j} \leq B, v_{i,q} \leq C, \forall c $ and $ j$, and $y_{i,i}=0 , v_{i,i} = 0, \forall i$. Note that the expression given in \eqref{t_iter} is applicable for any clustering structure as long as the CHs and BDs remove the redundant models that is reflected from the previous rounds.

\textbf{(3) Local learning process:} We consider the SGD algorithm to solve the local training problem of devices. Let $\mathbf w^{\tilde{t}}_{c_i,n}(t)$ denote the local model parameter of the $n$-th local device with the associated $c_i$-th CH at the $t$-th global round and the $\tilde{t}$-th local iteration. Define $\mathbf w^{*}_{c_i,n}(t)$ as the local model parameter  of the $n$-th local device with the associated $c_i$-th CH at the $t$-th global iteration after the local iterations are converged, which can be expressed as
\begin{equation}\label{LC1}
\mathbf w^*_{c_i,n}(t)=\arg \min_{\mathbf w_{c_i}} \nabla L_n(\mathbf w_{c_i}(t)).
\end{equation} 
The local model parameter $\mathbf w^{\tilde{t}}_{c_i,n}(t)$ in the $\tilde{t}$-th local iteration  is
updated according to the gradient of $L_n(\mathbf w_{c_i}(t))$ and the learning
rate $\delta$, i.e.,
\begin{equation}\label{LC2}
\mathbf w^{\tilde{t}+1}_{c_i,n}(t)=\mathbf w^{\tilde{t}}_{c_i,n}(t)-\delta\nabla L_n(\mathbf w^{\tilde{t}}_{c_i,n}(t)),
\end{equation} 
where $\mathbf w^{0}_{c_i,n}(t)=\mathbf w_{c_i}(t)$ because it is the initial value and can be obtained from the $c_i$-th CH. Then, the modified loss function of the $n$-th local device at the
$t$-th global round $\nabla G_n(\mathbf w^{\tilde{t}}_{c_i,n}(t))$ can be calculated as $\nabla G_n(\mathbf w^{\tilde{t}}_{c_i,n}(t))=\nabla L_n(\mathbf w^{\tilde{t}}_{c_i,n}(t))-\nabla L_n(\mathbf w_{c_i}(t))+\eta \nabla L(\mathbf w_{c_i}(t)),$ where $\eta$ is a positive constant to control the FL convergence rate. Since $\mathbf w^{*}_{c_i,n}(t)$  is the converged local model parameter, we have 
\begin{align}\label{updated}
\nabla G_n(\mathbf w^{*}_{c_i,n}(t))=& \nonumber \nabla L_n(\mathbf w^{*}_{c_i,n}(t))-\nabla L_n(\mathbf w_{c_i}(t))\\&+\eta \nabla L(\mathbf w_{c_i}(t))=0.
\end{align} 
The local iterations are continued until the local model accuracy $\epsilon_l$
is reached, which is defined as $G_n(\mathbf w^{\tilde{t}}_{c_i,n}(t))-G_n(\mathbf w^{*}_{c_i,n}(t))\leq \epsilon_l[G_n(\mathbf w_{c_i}(t))-G_n(\mathbf w^{*}_{c_i}(t))]$. On the other hand, the global iteration continues until the global model accuracy
$\epsilon_g$ is reached, which is defined as
$L(\mathbf w(t))-L(\mathbf w^{*})\leq \epsilon_g[L(\mathbf w^0)-L(\mathbf w^{*})]$.

	%\vspace{-0.8cm}
\section{Modeling and Problem Formulation}\label{PF}

We consider the energy consumption minimization problem that involves a joint management of computation and communication resources, CH clustering, and device
selection problem for the delay-constrained FL. As such, we evaluate the introduced FL scheme from communication and computation perspective.  Therefore, first we provide FL time and total energy consumption expressions, and then we formulate the energy consumption minimization problem.

\subsection{FL Time and Energy Consumption}

\textbf{ (1) FL time constraint:} At the $t$-th global iteration, the FL delay  is caused by: (i) the computation time for local model training at the scheduled local devices $\mathcal N_l$;  (ii) the
transmission time for transmitting local trained updates to the associated CHs $\mathcal C$; and (iii) the transmission time for transmitting the clusters aggregated models from the associated CHs to the scheduled devices $\mathcal N_l$. 

The $n$-th local device executes a local update algorithm until a local accuracy $\epsilon_l$ is achieved \cite{LC}. Let us denote $Q_n$ as the number of CPU cycles to process one data sample and $f_n$ is the computational frequency of the CPU in the $n$-th local device (in cycles per second). We obtain the number of CPU cycles of the $n$-th local device in a one local iteration over its dataset $\mathcal D_n$ as $Q_nD_n$ and the computation time for a one local iteration in the $n$-th device as $\frac{Q_n D_n}{f_n}$ \cite{CPU1, CPU2}. The computation time of the $n$-th local device is then expressed as
\begin{equation}
T^{comp}_n=T_l\frac{Q_nD_n}{f_n},
\end{equation}
where $T_l$ is the number of local iterations to reach the local accuracy $\epsilon_l$ in the $n$-th local device. The training parameters of the local devices are considered to have a fixed size \cite{S1}, and denote $s$ as parameter data size of the $n$-th local device. Therefore, the transmission time of the $n$-th local device for uploading
its parameters to the $c$-th CH on the $z$-th RRB is
$T^{com}_n=\frac{s}{R_{c,z}^n}$.

The aggregated model parameters can only be updated by the CHs
after all local model parameters are received from the scheduled devices and BDs. Consequently, the FL time is  dominated
by the longest duration time of receiving the local parameters from all scheduled devices and the longest duration time of aggregated models from the CHs to its scheduled devices. Note that the transmission duration of the $c$-th CH to transmit its aggregated parameter with size $s$ to its scheduled devices $\mathcal N_l^c$ is $T^{com}_c=\max_{n \in \mathcal N_l^c}\left\{\frac{s}{R^n_c}\right\}$. Hence, the learning time $\tau_c$ at the $c$-th CH can be calculated as
\begin{align}
\label{FL_1}
\tau_c&= \nonumber \max_{n \in \mathcal N_l^c}\{T^{comp}_n+T^{com}_n\}+T^{com}_c \\& =\max_{n \in \mathcal N^c_l}\left\{T_l\frac{Q_n D_n}{f_n}+\frac{s}{R_{c,z}^n}\right\}+\max_{n \in \mathcal N_l^c}\left\{\frac{s}{R^n_c}\right\}.
\end{align}

Therefore, for each global iteration, the total learning time $\tau=\max_{c\in \mathcal C}(\tau_c)$, which should be no more than the  QoS requirement
(i.e., $\tau$ should be no more than the maximum federated
learning time $T_\text{max}$). Typically, this constraint, over all global iterations $T$, is expressed as
\begin{align}\label{17}
T\left\{\max_{n \in \mathcal N^c_l}\left\{T_l\frac{Q_n D_n}{f_n}+\frac{s}{R_{c,z}^n}\right\}+\max_{n \in \mathcal N^c_l}\left\{\frac{s}{R^n_c}\right\}\right\}\leq T_\text{max},
\end{align}

\textbf{(2) Energy consumption:} The total energy consumption of our considered system is caused by the following:
\begin{itemize}
\item 	The energy that is consumed due to local computation at the local devices, where the energy 	consumption of the $n$-th local device to process a single CPU cycle is $\alpha f^2_n$, and $\alpha$	is a constant related to the switched capacitance \cite{CPU2, CPU3}. Thus, the energy consumption of the $n$-th local device	for local computation is $
E^{comp}_n=T_lQ_nD_n\alpha f^2_n$ \cite{Ansari1}.

\item The energy consumption to transmit local model parameters to the associated CHs can be denoted by $E^{com}_n$ and calculated as $p_nT^{com}_n$. Let us denote the total energy consummation at the $n$-th local device as $E_n=E^{comp}_n+E^{com}_n$. Similarly, the energy consumption to transmit the $c$-th CH aggregated model	 to the associated local devices is denoted by $E^{com}_c$ and calculated as $p_cT^{com}_c$.
\end{itemize}

Therefore, the total energy consumption of our envisioned system can be	calculated as 
\begin{align}\nonumber 
E&=T\left\{\sum_{n\in\mathcal N_l}E_n+\sum_{c\in\mathcal C}E^{com}_c\right\} \\ &=T\left\{\sum_{n\in\mathcal N_l}\left[T_lQ_nD_n\alpha f^2_n+ \frac{sp_n}{R^n_{c,z}}\right]+\sum_{c\in\mathcal C}\left[\frac{sp_c}{R_{c}}\right]\right\}.
\end{align}

\subsection{Problem Formulation}
	%\vspace{-0.3cm}
We consider the following two binary optimization variables: (i) $s_{n,c}$ for device-CH scheduling; and (ii) $r^n_{c,z}$ for device-RRB assignment. In addition, we consider the computation frequency allocation vector $\textbf f_\text{N}=[f_n]$.  
The energy consumption minimization optimization problem  is  formulated as
\begin{subequations}
	\begin{align} \nonumber 
	& \text{P0:} 
	\min_{\substack{s_{n,c} \in \{0,1\}, r^n_{c,z} \in \{0,1\},\\ \textbf{f}_{N}, \mathcal C\in \mathcal P(\mathcal N)}}  E\\
	&\rm s.t.
	\begin{cases}  \nonumber
	\hspace{0.2cm} \text{C1:}\hspace{0.2cm} \sum_{c\in \mathcal C}s_{c,n} = 1 ~\&~ \sum_{z\in \mathcal Z}r^n_{c,z} =1, \forall n \in \mathcal N_l, \\ 
	\hspace{0.2cm} \text{C2:}\hspace{0.2cm} \mathcal B^c\cap \mathcal B^{c'} = 1, \forall (c,c')\in \mathcal C,\\ 	\hspace{0.2cm} \text{C3:}\hspace{0.2cm} f^{\min}_n\leq f_n \leq f^{\max}_n, ~\forall n\in \mathcal{N}_l,\\	\hspace{0.2cm} \text{C4:}\hspace{0.2cm} \tau\leq T_\text{max},\\	
		%\hspace{0.2cm}	 \text{C5:}\hspace{0.2cm} s_{n,j}\in \{0,1\}, r^k_{i,j}\in \{0,1\}.
	\end{cases}
	\end{align}
\end{subequations}
In P0, C1 indicates that each local device is assigned to only one CH and to only one RRB in that CH; C2 indicates that at maximum one BD can be assigned to two adjacent CHs   at the same time. In practice and due to the nature of D2D communications, we can have more than one BD at the overlap regions of adjacent CHs. In this case, the BD in the middle of overlapped regions is selected and other BDs are scheduled to the most suitable CHs. C3 is the constraint on local computation resource allocation of devices; and C4 indicates the QoS requirement on the FL time.

To simplify C4 and using \eref{17}, we have  $T\max_{n \in \mathcal N^c_l}\left\{T_l\frac{Q_n D_n}{f_n}+\frac{s}{R_{c,z}^n}+\frac{s}{R^n_c}\right\}\leq T_\text{max}$  which can be transformed into $T\left(T_l\frac{Q_n D_n}{f_n}+\frac{s}{R_{c,z}^n}+\frac{s}{R^n_c}\right)\leq T_\text{max}, \forall n\in\mathcal N^c_l, c\in \mathcal C$. Hence,  $T_l\frac{Q_n D_n}{f_n}+\frac{s}{R_{c,z}^n}\leq T_\text{c,max}$, where $T_\text{c,max}=\frac{T_\text{max}}{T}-\frac{s}{R^n_c}$. Hence, similar to \cite{Ansari1}, the
lower bound of device's computation frequency can
be calculated as $f_n\geq \frac{T_lQ_n D_n}{T_\text{c,max}-\frac{s}{R_{c,z}^{n}}}$. For simplicity, we denote $\hat f_n = \frac{T_lQ_n D_n}{T_\text{c,max}-\frac{s}{R_{c,z}^{n}}}$. Then, $f_n$ satisfies  $f_n\geq \hat f_n$, and accordingly, P0 can  be transformed into
\begin{subequations}
	\begin{align} \nonumber 
	& \text{P1:} 
	\min_{\substack{s_{n,c} \in \{0,1\}, r^n_{c,z} \in \{0,1\},\\ \textbf{f}_{N}, \mathcal C\in \mathcal P(\mathcal N)}}  E\\
	&\rm s.t.
	\begin{cases}  \nonumber
	\hspace{0.2cm} \text{C1}, \text{C2,}\\
	\hspace{0.2cm} \text{C3:}\hspace{0.2cm} \hat f_n\leq f_n \leq f^{\max}_n, ~\forall n\in \mathcal{N}_l,\\	\hspace{0.2cm} \text{C4:}\hspace{0.2cm} \tau=T\left\{\max_{n \in \mathcal N^c_l}\left\{T_l\frac{Q_n D_n}{f_n}+\frac{s}{R_{c,z}^n}+\frac{s}{R^n_c}\right\}\right\}. %\hspace{0.2cm} \text{C5}.
	\end{cases}
	\end{align}
\end{subequations}
	
\section{Proposed Solution}\label{J}
	
\subsection{Proposed Solution Approach}
	
Problem \text{P1}  is a mixed-integer non-linear programming problem as it involves a joint optimization of the continuous computation frequency allocation and discrete resource assignment variables (e.g., CH-devices/BD and RRB scheduling). In addition, the selection of CHs $\mathcal C$ and their scheduled local devices/BDs are inter-dependent and should be considered jointly. We can readily show that \text{P1} is an NP-hard problem. Even if we fix one of these variables, \text{P1} still remains NP-hard. An exhaustive search is required to obtain the global optimal solution of problem \text{P1}, which is infeasible in practical systems. To overcome the aforementioned computational intractability, an iterative approach is devised, where we decompose \text{P1} into two sub-problems. Particularly, considering that $f^*_n~ \forall n\in \mathcal N_l$ are given, the energy consumption can be obtained from the following optimization problem.
\begin{align} 
	 \text{P2:}
	\min_{\substack{s_{n,c} \in \{0,1\}, r^n_{c,z} \in \{0,1\},\\ \mathcal C\in \mathcal P(\mathcal N)}}  E~~
	\rm s.t.
		\hspace{0.2cm} \text{C1}, \text{C2,} ~\text{C4}.
\end{align}
While for the given CH clustering and device-CH/RRB scheduling, the computation frequency allocation of the local devices are
obtained by solving the following sub-problem
\begin{align} 
\text{P3:}
\min_{\substack{\textbf{f}_{N}}}  E~~
\rm s.t.
\hspace{0.2cm} \text{C3}, \text{C4}.
\end{align}
A suitable solution to \text{P1} is obtained by iteratively solving the
sub-problems \text{P2} and \text{P3}. The sub-problems \text{P2} and \text{P3} are solved in \sref{S1} and \sref{S2}, respectively, and the overall algorithm to solve \text{P1} is provided in \sref{A1}.

\vspace{-0.1cm}
\subsection{CH, BD, and Local Device Scheduling: P2 Reformulation}
\vspace{-0.1cm}
It is recalled that the set of CHs is $\mathcal C \in \mathcal P(\mathcal N)$. Let $\kappa_c(\mathcal C)$ be the
set of scheduled local devices that are associated to CH $c\in \mathcal C$. Note that $\kappa_c(\mathcal C)$ is a function of the set of CHs $\mathcal C$. However, for convenience, the set will be dropped unless it is required, e.g., we will use $\kappa_c$ instead of  $\kappa_c(\mathcal C)$.

%\textit{Remark 1:} Since devices have limited resources (e.g., computational capability and battery), we consider that devices can perform either local learning on their local datasets or aggregation and cannot perform both learning and aggregation.

\text{P2} considers a joint optimization of both
CHs selection and their scheduled devices and BDs, denoted
by $\mathcal C \in \mathcal P(\mathcal N)$, $\mathcal N_l$, and $\mathcal B$, respectively. In particular, \text{P2} involves  inner and outer optimization problems: the inner one optimizes the device/BD scheduling to reduce
the network utility function while the outer one aims at finding the most suitable CHs. Therefore, the joint inner and outer optimization sub-problem P2, for a given $f^*_n~ \forall n\in \mathcal N_l$, can be given at the top of the next page.
\begin{table*}
\begin{align} \label{p3a}
\text{P2.1:} &\min_{\substack{{\mathcal C \in \mathcal P(\mathcal N), \kappa_c(\mathcal C)\in \mathcal P(\mathcal A_c)}}} 
\sum_{c\in \mathcal C}\left( \sum_{n\in \kappa^*_c(\mathcal C)} E_n+E^{com}_c\right),\\
&\rm s. ~t. ~~
\kappa^*_c(\mathcal C)=\arg \min_{\mathtt \kappa_c(\mathcal C)\in \mathcal P(\mathcal A_c)}\left(\sum_{n\in \kappa_c(\mathcal C)}\left( T_l\frac{Q_n D_n}{f_n}+\frac{s}{R_{c,z}^n}\right)\right), \forall c \in \mathcal C, \label{p3p}  \\ \nonumber
&  \tau=T\left\{\max_{n \in \mathcal \kappa_c(\mathcal C)}\left\{T_l\frac{Q_n D_n}{f^*_n}+\frac{s}{R_{c,z}^n}+\frac{s}{R^n_c}\right\}\right\}.
\end{align}
\hrulefill
\vspace*{-0.5cm}
\end{table*}

Unfortunately, the outer optimization problem \eref{p3a} and the inner  optimization problem \eref{p3p} of P2.1 are inter-dependent and should be solved jointly. Solving P2.1 optimally using an exhaustive search is intractable \cite{D2D6, D2D7}. In this work, we develop a tractable graph theoretical and clustering method that carefully selects the number of clusters, their corresponding CHs, and scheduled local devices/BDs. As such, it achieves an effective solution to the optimization problem P2.1.

\subsection{CH Clustering and Local Device/BD Scheduling: Solution to P2.1}\label{S1}
	
\textbf{(1) Graph representation:} We first design a conflict graph that considers the restrictions in device-CH/RRB scheduling and the partial connectivity
of the CHs. In this context, a conflict graph is constructed for each device $c$ in the network and denoted by $\mathcal G_c(\mathcal V_c, \mathcal E_c)$ wherein $\mathcal V_c$ and $\mathcal E_c$ are the set of vertices and edges, respectively. Since each device $c$ can schedule local devices in its coverage zone $\mathcal A_c$ only, a vertex $v^z_{c,n} \in \mathcal V_c$ is generated for each local device $n\in \mathcal A_c$ and for each RRB $z$. In the conflict graph, two distinct vertices $v^{z}_{c,n}$ and $v^{z'}_{c,n'}$ are adjacent by a scheduling conflict edge if one of the
following cases occurs:
\begin{itemize}
	\item \textbf{CC1:} The same  local device is associated with both vertices $v^{z}_{c,n}$ and $v^{z'}_{c,n'}$, i.e., $n=n'$.
	\item \textbf{CC2:} The same RRB is associated with both vertices $v^{z}_{c,n}$ and $v^{z'}_{c,n'}$, i.e., $z=z'$ and $n \neq n'$.
\end{itemize} 
Recall that each selected CH is granted a limited number of $Z$ RRBs. Thus, each CH can schedule at maximum $Z$ local devices, and the maximum of $Z$ vertices can be represented by the minimum-weight independent set (MWIS) in the conflict graph. In what follows, we transform problem \eref{p3p}  into a MWIS problem. First, we present the following remarks about the MWIS. 
\begin{itemize}
	\item  Any independent set (IS) $\kappa_c$ in the conflict graph $\mathcal G_c$ must satisfy: i) $\kappa_c\subseteq \mathcal G_c$; ii) $\forall v, v' \in \kappa_c$, where $(v, v') \notin \mathcal E_c$.
	\item A minimal IS in an undirected graph cannot be expanded to	add one more vertex without affecting the pairwise non-adjacent vertices.
	\item The IS $\kappa_c$  is referred to an MWIS of $\mathcal G_\text{c}$ if it satisfies: i) $\kappa_c$ is an IS in graph $\mathcal G_\text{c}$; ii) the sum weights of the vertices in $\kappa_c$ offers the minimum among all ISs of $\mathcal G_\text{c}$. Therefore, the MWIS will be denoted as $\kappa^*_c$.
\end{itemize}
Given $\mathcal{G}_c$,  the optimization problem \eref{p3p} for a fixed $f^*_n~ \forall n\in \mathcal N_l$ is similar to
MWIS problems. In  MWIS problems, two vertices must be non-adjacent in the conflict graph, which is similar to problem \eref{p3p}, where two local devices cannot be allocated with the same RRB. Likewise, the same RRB cannot be allocated to two different local devices. Furthermore, the objective of \eref{p3p} is to minimize the energy consumption, while the
goal of MWIS is to select a set of vertices with minimum weights. Hence, the following corollary characterizes the solution of the optimization problem \eref{p3p}.

\begin{corollary}
The solution 
$\kappa^*_c(\mathcal C)$ of the optimization
problem \eref{p3p} for device $c$ is the MWIS $\kappa^*_c$
in the conflict graph $\mathcal G_c$ of device $c$ in which the weight of	each vertex $v^z_{c,n}$ is $w(v^z_{c,n})=T_l\frac{Q_n D_n}{f^*_n}+\frac{s}{R_{c,z}^n}$. %The solution of the MWIS problem is presented in Algorithm \ref{alg:LGS}.
\end{corollary}

\textbf{(2) CH Clustering:}
Essentially, two adjacent clusters are formed, and accordingly, their corresponding CHs can simultaneously transmit the aggregated models to the scheduled local devices if the CHs' coverage zones are overlapped by at least one BD.  This is to ensure that the shared BD can exchange the received aggregated models to the adjacent CHs. 

We construct the whole set of clusters $\mathbf F$ such that: (i) all coverage zones of adjacent clusters $F \in \mathbf F$ are pairwise overlapped by one BD; and (ii) within the same cluster $F$, each subset of local devices is interfering with at least another local device. Mathematically, the clusters in $\mathbf F$ should satisfy: (i) $\mathbf F= \{F  \in \mathcal P(\mathcal N) ~| \mathcal A_c\cap \mathcal A^{tot}(F\backslash c)\neq \emptyset, \forall c\in  F\}$; (ii) $\mathcal A^{tot}(F)\cap \mathcal A^{tot}(F') =1, \forall (F, F') \in \mathbf F$.
Therefore, we have the following lemma for the overlapped clustering construction. 

\begin{lemma}
For any particular set $\mathcal C$, the corresponding overlapped clustering $\mathcal F \in \mathbf F$ satisfies the following constraints
	\begin{subequations}
		\begin{align} \label{eq2a}
		&\bigoplus_{F\in \mathcal F} F= \mathcal C\\ \label{eq2b}
		& \mathcal A^{tot}(F)\cap \mathcal A^{tot}(F')= 1, \forall F\neq F'\in \mathcal F,  \\ \label{eq2c}
		&\mathcal A_c\cap \mathcal A^{tot}(F\backslash c) \neq \emptyset, \forall c \in F, 
		\end{align}
\end{subequations}
where $\mathcal A^{tot}(F)$ is the total coverage zone of all local devices in the cluster $F$ defined as $\mathcal A^{tot}(F)=\cup_{c\in F}\mathcal A_c$. Thus, for any
$\mathcal C$, there exists a clustering $\mathcal F$ satisfying the constraints
\eref{eq2a}, \eref{eq2b}, and \eref{eq2c} simultaneously. In \eref{eq2a}, $\bigoplus$ denotes a partition symbol.
\end{lemma}

\begin{proof} 
The proof is analogous to \cite[Lemma 1]{D2D7}, and is omitted due to the space limitation.
%We can prove that $\mathcal F$ associated with $\mathcal N_c$ satisfies \eref{eq2a}-\eref{eq2c} by performing a sequential clustering method as follows. Firstly, we generate a cluster $F$ by picking up an arbitrary element $c$ from the set $\mathcal N_c$, and we repetitively add all the remaining devices from $\mathcal{N}_c$ that are interfering with $F$ to it. We stop until there is at least one device remaining in $\mathcal N_c$ that interferes with $F$, which is denoted by $d$. Consequently, constraint \eref{eq2c} holds. It is clear that the remaining devices in $\mathcal N_c$ are interfering with $F$ by a device $d$.  Secondly, we generate a new cluster $F'$ by picking up one of the remaining CHs in $\mathcal N_c$ that is overlapped with CH $c$ by device $d$, then we add all the interfering devices to CH $c$. Notably, the devices in $F$ interfere with the devices in $F'$, and accordingly, $F$ and $F'$
%are interfering by at least one device, i.e., $\mathcal F$ satisfies
%\eref{eq2b}. We repeat this construction process until $\mathcal N_c = \emptyset$ which ensures that constraint \eref{eq2a} holds. As a conclusion, $\mathcal F$ satisfies \eref{eq2a}-\eref{eq2c} for any combination of CHs $\mathcal N_c$. 	
\end{proof}

It is recalled that local devices can only communicate with their respective neighbors over a wireless D2D link whose connectivity can be characterized by an undirected graph $\mathcal G(\mathcal V, \mathcal E)$ with $\mathcal V$ denoting the set of vertices and $\mathcal E$ the set of edges. To design a scheduling scheme meeting the aforementioned clustering properties, we construct the graph $\mathcal G(\mathcal V, \mathcal E)$ as follows. We represent a cluster $F$ by a vertex $V$ in the graph, and two vertices are connected by a scheduling link if the represented clusters are satisfying  \eref{eq2b}. Next, to select the device-CH/RRB scheduling that provides a minimum energy consumption while maximizing the number of scheduled local devices in each cluster and guarantees the FL time constraint, we assign a weight
$w(V)$ to each vertex $V \in \mathcal V$. For notation simplicity, we define the utility
of the $n$-th local device as $E_n=T_l\frac{Q_n D_n}{f^*_n}+\frac{s}{R_{c,z}^n}$ and the utility of cluster $F$ represented by vertex $V$ as $w(V)= \sum_{n\in \kappa^*_c} E_n+E^{com}_c$. Therefore, the following theorem characterizes the
solution to the energy consumption minimization problem P2.1, for a given $f^*_n~ \forall n \in \mathcal N_l$, in a partially connected D2D network.

\begin{theorem}\label{th:1w}
The  solution to the energy consumption minimization problem P2.1 in the partially connected D2D network is equivalent to the weighting-search method, in which the weight of each
vertex $V$ representing cluster $F$ is
\begin{align} \label{eqweight} 
w(V)=\sum_{n\in \kappa^*_c} E_n+E^{com}_c,
\end{align}
where $\kappa^*_c$ is the MWIS problem in the conflict graph of CH $k$ for device scheduling
\begin{align} \label{eqwww}
\kappa^*_c=\arg \min_{\mathtt \kappa_c\in \mathcal G_c}\left(\sum_{n\in \kappa_c}\left( T_l\frac{Q_n D_n}{f^*_n}+\frac{s}{R_{c,z}^n}\right)\right).
\end{align}
\end{theorem}

Essentially, any possible solution to P2.1 using Theorem 1, representing CHs selection and their scheduled local devices, satisfies the following criteria: (i) the represented local devices by vertices in the conflict graph are scheduled to the best CHs/RRBs such that they transmit their local parameters quickly; (ii)	 the represented CH in each vertex is selected, so as it has a good reachability to many local devices in its coverage as well as it is overlapped with the neighboring CHs. This does not only increase the number of selected local devices that can participate in the training process, but also minimizes the number of clusters; and (iii)	 a vertex (representing by a subset of local devices) with a smaller weight (representing a sum of their consumed energy) will have a higher priority to be	selected in participating in the training process to minimize	the objective of P2.1.

\begin{algorithm}[t!]
	\SetAlgoLined
	\KwData{$\mathcal{N}, \mathcal{Z}$, $P_n, H^n_{c,z}$, and $f_n$,  $(n,c,z)\in\mathcal{N}\times\mathcal{N}\times\mathcal{Z}$}
	%		\STATE \textbf{Repeat:}\;
	
	\textbf{Initialize:} $\Gamma^* = \emptyset$ and $\mathcal G=\emptyset$;\;

	Generate a set of clusters $\mathbf F$ according to Section IV-C as follows.\;
	{\begin{itemize}
		\item $\mathbf F= \{F  \in \mathcal P(\mathcal N) ~| \mathcal A_c\cap \mathcal A^{tot}(F\backslash c)\neq \emptyset, \forall c\in  F\}$;\;
		\item $\mathcal A^{tot}(F)\cap \mathcal A^{tot}(F') =1, \forall (F, F') \in \mathbf F$;\;
	\end{itemize}}
	 Construct the graph $\mathcal G$ as explained in IV-C.\;

	\For{$V \in \mathcal G$}{
	    \For{$c=1,\cdots, |F|$}{
		construct the conflict graph $\mathcal G_c(\mathcal V_c, \mathcal E_c)$ for local device scheduling;\;
		$\forall v^z_{c,n} \in \mathcal G_c$, calculate the weight $w(v^z_{c,n})=T_l\frac{Q_n D_n}{f^*_n}+\frac{s}{R_{c,z}^n}$;\;
		solve the MWIS selection problem in Corollary 1 as follows:\; 
		    initialize $\kappa_c= \emptyset$ and		set $\mathcal G_c(\kappa_c) \leftarrow \mathcal G_c$;\;
		\While{$\mathcal G_c(\kappa) \neq \emptyset$}{
				 Select $v^*=\arg\min_{v^z_{c,n}\in \mathcal G_c(\kappa_c)} \{w(v^*)\}$;\;
	             Set $\kappa \leftarrow \kappa_c \cup v^*$ and		 obtain $\mathcal G_c(\kappa_c)$;\;
	                                               }
	Set $\kappa^*_c \leftarrow \kappa_c$;\;
		}
		calculate $w(V)$ using \eref{eqweight};\;
	   }
	Update $\mathcal G(\Gamma^* ) \leftarrow \mathcal G$;\;
	\Repeat{$\mathcal G(\Gamma^* ) = \emptyset$}{
		$V^*=\arg\min_{V\in \mathcal G(\Gamma )} \{w (V)\}$;\; 
			    construct a set $\mathcal G(V^*)$ of vertices connected to 
	  $V$;\;
		set $\Gamma^* \leftarrow \Gamma^* \cup V^*$ and 	set $\mathcal G(\Gamma^* ) \leftarrow \mathcal G(V^*)$;\;
	}
		\KwResult{$\Gamma^*$}

	\caption{CH clustering and local device/BD scheduling over the scheduled RRBs} \label{alg1}
\end{algorithm}

The overall steps of cluster generation, conflict graph construction, and MWIS selection to obtain the CH clustering and device scheduling over the scheduled RRBs are summarized
in Algorithm \ref{alg1}.

%\vspace{-0.55cm}
\subsection{Computation Frequency Allocation: Solution  to Sub-problem P3}\label{S2}
%\vspace{-0.1cm}
For the resulting selected CHs and their scheduled devices from solving P2.1, sub-problem P3  is reduced to the following sub-problem
	\begin{align} \nonumber 
	 \text{P4:}
	&\min_{\substack{f_n}} \sum_{n\in\mathcal N_l} T_lQ_nD_n\alpha f^2_n\\&
	\rm s.t.
	\hspace{0.2cm} \max\left\{f^{\min}_n,\hat f_n\right\} \leq f_n \leq f^{\max}_n, ~\forall n\in \mathcal{N}_l.
	\end{align}
\begin{lemma}
The closed-form solution of sub-problem P4 is obtained as
\begin{equation}
\label{close_form_P5}
\begin{split}
f_n=\begin{cases}
& f^{\min}_n,  ~\text{if} ~\hat f_n \leq f^{\min}_n \\
& \hat f_n , ~\text{if} ~ f^{\min}_n< \hat f_n <  f^{\max}_n\\
&  f^{\max}_n, ~\text{if} ~ \hat f_n  \geq  f^{\max}_n,
\end{cases}
\end{split}
\end{equation}
\end{lemma}

\begin{proof}  The proof of the closed-form procedure is analogous to \cite{C5, Dual1}, and  is omitted due to the space limitation.\ignore{The objective function in P4 is monotonically
increasing with respect to $f_n$ when $f_n \geq 0$ \cite{C5}. To minimize the objective function, $f_n$ should be in the
feasible set of $\left\{\max\left\{f^{\min}_n,\hat f_n\right\}, f^{\max}_n\right\}$. Hence, the closed-form
solution is $f_n=\min\left\{\max\left\{f^{\min}_n,\hat f_n\right\}, f^{\max}_n\right\}$. Then, we repetitively perform the following two closed-form procedures \cite{C5, Dual1}.
\begin{itemize}
	\item If $\frac{T_lC_n D_n}{T_{c,max}-\frac{d_n}{R_{c,z}^{n}}}< f^{\max}_{n}$, the local computation of the $n$-th device is feasible. Thus, we set $f_n=\max\left\{f^{\min}_n, \frac{T_lC_n D_n}{T_{c,max}-\frac{d_n}{R_{c,z}^{n}}}\right\}$.
	\item If $\frac{T_lC_n D_n}{T_{c,max}-\frac{d_n}{R_{c,z}^{n}}} = f^{\max}_{n}$, the local computation of the $n$-th device is feasible. Thus, we set $f_n=f^{\max}_{n}$. 
	
\end{itemize}}
\end{proof}

%\vspace{-0.55cm}
\section{Development and Properties Of The Proposed FL-EOCD Algorithm}\label{A}
%\vspace{-0.1cm}

\subsection{Overview of the FL-EOCD Algorithm}\label{A1}
%\vspace{-0.1cm}
The overall steps followed by the proposed FL-EOCD algorithm to obtain a suitable solution to problem P0 are summarized in Algorithm \ref{alg2}. FL-EOCD has two stages, namely, Stage-I and Stage-II. 

\textbf{(1) Stage-I:} The required steps to implement Stage-I are provided as follows. At the beginning, the BS obtains the
locations of the devices. Thereafter, based on the
current computation frequency allocation, the BS designs the FL local graph of each device and judiciously finds the suitable CH clustering and the local device scheduling. Specifically, in order to determine the CHs and schedule the devices, the BS constructs the clustering and the corresponding graph, and then  performs the greedy vertex-weighting algorithm \cite{MWIS1, MWIS2}.  First, the BS selects the vertex that has the following features: (i) it has a minimum weight in terms of energy consumption and (ii) it is connected to many neighbors that have minimum weights.  Next, it constructs the set of vertices that are connected to the selected vertex, and then finds the second minimum-weight vertex. This process is repeated until no neighbors are available in the designed clustering, and this is summarized in lines 21-25 of Algorithm \ref{alg1}. Meanwhile, the BS calculates the FL time of the scheduled devices according to (17) to check the feasibility of the scheduled local devices. Next, for each scheduled device, the BS updates the variable $f_n$. The aforementioned steps are iteratively repeated until the computation frequency allocation constraints over all the scheduled local devices are satisfied or the maximum number of iterations is reached. Upon convergence in Stage-I (i.e., lines 3-7 of Algorithm \ref{alg2}), the set of CH clustering, device/BD scheduling, and  computation frequency allocation among the scheduled local devices are obtained.

\textbf{(1) Stage-II:} To efficiently disseminate the local cluster aggregated models and perform FL of the CHs obtained by Stage-I, Stage-II is introduced in the FL-EOCD algorithm. In this stage, three main tasks are performed: (i) local learning at the associated devices;  (ii) models combination and forwarding at the BDs; and (iii) cluster model aggregations at the CHs. 
The obtained CHs from Stage-1 transmit their initial models $\mathbf w_{c_i}(t)$ to the associated devices and BDs. Afterword, the associated devices execute the local learning algorithm to update their local optimal models. Meanwhile, the BDs receive the updated models from the CHs and calculate the summation of the received aggregated models.  The CHs, consequently, collect all the local models from the associated devices and the models from the BDs to update their cluster models according to \eref{L_c}. The  aforementioned two stages are repeated until
the maximum number of global iterations is reached. Upon convergence
in Stage-II, the final global model is obtained.

%\textit{Remark 3:} Despite the convergence, due to the approach of updating the CH clustering, device scheduling, and computation frequency allocation variables in Stage-1, the
%final output of FL-ERCG is essentially a local optimal solution
%to problem $\mathcal P_2$. Nevertheless, problem $\mathcal P_2$ is NP-hard and the
%required computational complexity for obtaining the global
%optimal solution cannot be supported by practical dense systems.
%Furthermore, the simulation results depict that the
%performance gap in terms of global convergence between the solution obtained using FL-ERCG
%and the centralized BS solution is
%insignificant. Therefore, FL-ERCG is efficient.

\begin{algorithm}[t!]
	\SetAlgoLined
	 \KwData{Length of training $L$, number of cluster's
	aggregations $T$, number of local
	iterations $T_l$, $K_{max}$, $\mathcal{N}, \mathcal{Z}$, $P_n, H^n_{c,z}$, and $f_n$, $(n,c,z)\in\mathcal{N}\times\mathcal{N}\times\mathcal{Z}$}
	\textbf{Initialize:} The iteration index $t=1, k=1$, $f^{(0)}_n=f^{\min}_n$,  $f^{(k)}_n=f^{\max}_n, \forall n\in \mathcal N$, the initial global model of CH $c_i$ $\mathbf w_{c_i}(t-1)$ and broadcast it among the local devices;\;
	\Repeat{$t=T$}{
		\KwData{Locations of the devices and $H^n_{c,z}$ ~~~~~~~~~~~~~~(Start of Stage-I)}
		\Repeat{$f^{(k)}_n \neq f^{(k-1)}_n ~and~ k< K_{\max}$}{
			by plugging the updated $f^{(k)}_n, \forall n\in \mathcal N_l$ to Algorithm
			\ref{alg1}, update the CH clustering and local device scheduling;
			
			for the updated CH clustring and device scheduling, update the solution $f^{(k)}_n$ of the 
			problem P4 according to Lemma 1;\; 
			update $k=k+1$;
		}
		 \KwResult{$\mathcal C$, $\mathcal N^{c_i}_l$, and $f^{*}_n$, $\forall n \in \mathcal N^{c_i}_l, c_i\in \mathcal C$; ~~~~~~~~~~~~(End of Stage-I)}
		\For{$c_i=1:C$}{
			the $c_i$-th CH broadcasts $\mathbf w_{c_i}(t-1)$ to its scheduled local devices and BDs   (Start of Stage-II);\; 
			each BD $b_{i, {j}}$ receives $\mathbf{w}_{c_i}(t-1)$ and $\mathbf{w}_{c_{j}}(t-1)$ from the adjacent CHs $c_i,c_{j} \in \mathcal B_{b_{i,j}}$, performs the local learning, calculates the summation of the received models, and aggregates the resulting sum with its own model  into  $\mathbf{w}^{ag}_{b_{i,j}}(t)$ using \eref{BD_agg};\;
			\While{$\tilde{t}<T_l$}{
				each device $n \in \mathcal N^{c_i}_l$ performs SGD to compute $L_n(\mathbf{w}_{c_i}(t-1))$ using \eref{L_n};\; 
				each local device $n \in \mathcal N^{c_i}_l$ updates its model as $\mathbf{w}^{\tilde{t}+1}_{c_{i},n}(t)=\mathbf{w}^{\tilde{t}}_{c_i,n}(t-1)-\delta\nabla L_n(\mathbf{w}^{\tilde{t}}_{c_i,n}(t-1))$;\;

			}
			$\forall n\in \mathcal N^{c_i}_l$, device $n$ transmits $\mathbf{w}^*_n(t)$ to the scheduled CH;\;
			the $c_i$-th CH computes  $\mathbf{w}^{ag}_{c_i}(t+1)$ based on \eref{wcag} using $\mathbf{w}^*_{c_i,n}(t), \forall n\in \mathcal N^{c_i}_l$, $\mathbf{w}^{ag}_{b_{i,j}}(t)$, and $\mathbf w^{ag}_{c_i}(t)$;\;
		}
		update $\mathbf w_{c_i}(t-1)=\mathbf w_{c_i}(t)$ and $t=t+1$; ~~~~~~~~~~~~~~~~~~~~~~~~~~~~~~~~~~(End of Stage-II)\;
					}	
				\KwResult{final global model $\mathbf w$}
				
	\caption{FL-EOCD Algorithm} \label{alg2}
	
\end{algorithm}

%\vspace{-0.55cm}
\subsection{Computational and Signaling Cost of FL-EOCD}
%\vspace{-0.1cm}
\textbf{(1) Computational Complexity:} We first analyze the computational complexity of Stage-I of the FL-EOCD algorithm. The required computational complexity of Stage-I is dominated by the complexity of constructing the cooperation clusters. This needs a computational complexity of $\mathcal O(|\mathbf F|)$. Then, checking the neighbors of all these clusters requires a computational complexity of  $\mathcal O(|\mathbf F|^2)$ \cite{D2D6, D2D7}. On the other hand, designing local graphs and finding the corresponding solution requires a computational complexity of $\mathcal O(NN_{avg}+NZ)$, where $N_{avg}$ is the average number of connected devices to any device in the network. Therefore, the overall computational complexity of Stage-I of the FL-EOCD algorithm  is therefore obtained as $\mathcal O(|\mathbf F|^2+NN_{avg}+N Z)$. Thus, FL-EOCD requires a polynomial computational complexity to obtain a converged solution.

\textbf{(2) Signaling Overhead:} %The information exchanges required by the FL-EOCD algorithm are summarized as follows.
At each iteration of the FL-EOCD algorithm, the BS executes Algorithm \ref{alg1} that has a required signaling overhead as follows.  At each iteration of Algorithm \ref{alg1}, devices first transmit their channel state information (CSI) and the indices of devices in their coverage zones to the BS. In this phase, at most $2N$ information exchanges are required.  Subsequently, the BS executes Algorithm \ref{alg1} and transmits the resulting device scheduling and BD selections to the CHs, which needs at most $N_c$ information exchanges.  In this phase, a total of $2N+N_c$ information exchanges are required for Algorithm 1. Until the convergence of Algorithm \ref{alg1},  a total number of the required information exchanges is obtained as $K_{\max}(2N+N_c)$, where $K_{max}$ is the maximum number of iterations for solving P2 and P4.  Each device needs to transmit its model parameter to the scheduled CH, and this requires $N_cN^c_l$ information exchanges for all the clusters. In addition, CHs need to broadcast their global models to the scheduled devices and BDs, which requires $N_c$ information exchanges. Therefore, a total of  $T(K_{\max}(2N+N_c)+N_cN^c_l+N_c)$  information exchanges are required between the devices and the BS during the execution of the FL-EOCD algorithm.

%\vspace{-0.4cm}
\section{Numerical Results}\label{NR}
%\vspace{-0.1cm}
\subsection{Simulation Setting and Schemes Under Consideration}
%\vspace{-0.1cm}
(1) \textbf{Network configuration:} For the numerical performance evaluations, we consider a circular network area having a radius of $900$ m with one BS at its center and multiple randomly distributed devices. The channel model for D2D transmissions follows the standard model \cite{D2D4}, which consists of three components: 1) path-loss of $148+40\log_{10}(\text{dis.[km]})$; 2) log-normal shadowing with $4$ dB standard deviation; and 3) Rayleigh channel fading with zero-mean and unit variance. On the other hand, the path-loss of the cellular is  $128.1+37.6\log_{10}(\text{dis.[km]})$ \cite{RRB1, RRB2}. We consider that the channels are perfectly estimated. The noise power and maximum local device (BD and CH) power are assumed to be $-174$ dBm/Hz and $p_n=p_c=1$ W, respectively. The transmit power of the BS is set to $3$ W. The bandwidth of each RRB is $2$ MHz. Unless otherwise stated, we set the numbers of devices to $26$ and we consider $22$ RRBs in the system. The
remaining simulation parameters that are selected based on \cite{C5, Ansari1} are summarized in Table \ref{table_1}. In the ensuing simulation
results, each simulation experiment is repeated $200$ times with randomly generated devices' locations, and the average
results over all the considered simulation instances are presented. We consider a deep learning image classification problem over MNIST dataset \cite{semi}. The considered dataset contains $60$ K images, where each image is one of $10$ labels. We divide the
dataset into the devices’ local data $\mathcal D_n$ with  non-i.i.d. data heterogeneity, where each local dataset contains datapoints from two of the $10$ labels. In each case, $\mathcal D_n$ is selected randomly from the full dataset of labels assigned to $n$-th device. We also assume non-iid-clustering, where the maximum number of assigned classes for each cluster is $6$ classes. For ML models, we use a deep neural network with $3$ convolutional layers and $1$ fully connected layer. The total number of trainable parameters is $9098$.

(3) \textbf{Schemes Under Consideration:}
To showcase the effectiveness of the proposed FL-EOCD scheme in terms of energy consumption and latency, we consider the following FL benchmark schemes.
\begin{itemize}
	\item \textbf{Star-based FL:}  Here, instead of local cluster aggregations at the CHs, the scheduled devices upload their local models directly to the BS for a global model aggregation. At the BS, an averaging FL scheme is applied to calculate the global model to be transmitted to the scheduled devices for initiating the next iteration of local learning. For a fair comparison, the performance of the centralized FL is optimized by leveraging the FL-EOCD framework. In particular, Algorithm \ref{alg1} is used to obtain the devices-RRBs scheduling and computation frequency allocation to minimize the energy consumption. 
	\item \textbf{Hierarchical FL:} Here, instead of transmitting the local models from scheduled devices directly to the BS, the selected CHs collect the local models from their devices and forward them to the BS for a global aggregation. For a fair comparison, the performance of the hierarchical FL is optimized by leveraging the FL-EOCD framework. In particular, we optimize the CH selection, device scheduling, and computation frequency allocation. 
\end{itemize}

\begin{table}[t!]\small
	\renewcommand{\arraystretch}{0.9}
	\caption{Simulation Parameters}
	\label{table_1}
	\centering
	\begin{tabular}{|p{5.0cm}| p{2.8cm}|}
		\hline
		
		\textbf{Parameter} & \textbf{Value}\\
		\hline 
		Circle radius of device’s service area, $\mathtt R$ & $400$ m\\
		\hline
		Local and cluster aggregated parameters size, $s$ & $9.1$ KB\\
		\hline
		%Device data size, $\mathtt D_n$ & $[0.5-1]$ Mbit\\
		%\hline
		Device processing density, $Q_n$ & $[400-600]$\\
		\hline
		Device computation frequency, $f_{n}$ & $[0.0003-1]$ G cycles/s\\
		\hline
		CPU architecture based parameter, $\alpha$ & $10^{-28}$\\
		\hline
	     FL time threshold $T_{max}$ & $1$ Second \\
	     \hline
	     Number of data samples, $D_n$ & $200$\\
	     \hline
	     %\textcolor{blue}{Loss function & $\gamma=2$ strongly convex and $\beta=4$ smooth}\\
	     %\hline 
	     %\textcolor{blue}{Local learning step size,} $\delta$ & $1/4$\\
	     %\hline 
	\end{tabular}
%\vspace{-1.8em}
\end{table}
%\vspace{-0.4cm}
\subsection{Performance of the Proposed FL-EOCD and Benchmark Schemes for Studying the Accuracy and Number of Iterations for Different Overlapped Clustering}

%In Fig. 9, we plot the convergence rate of the proposed FL-EOCD and benchmark schemes versus the number of iterations for different configurations of number of devices $N$ for a network of $30$ RRBs.

In Fig. \ref{fig4}, we plot the accuracy of the proposed FL-EOCD and benchmark schemes versus the number of global iterations for a network of $22$ RRBs and $26$ devices which are grouped into $3$ overlapped clusters. In the considered star-based and hierarchical FL schemes, the BS can receive  the local trained models of the devices, where each scheduled device transmits its trained model directly to the BS (in case of star-based FL) or through CHs (in case of hierarchical FL). As a result, the BS can aggregate all the local models and reaches accuracy of $88\%$ and $86\%$ with around $20$ global iterations for star-based and hierarchical FL schemes, respectively. In the considered decentralized system, at the beginning, each selected CH has the knowledge of the cluster model that depends on its scheduled devices in the corresponding cluster. Consequently, FL-EOCD has a low accuracy of $82\%$ at $20$ iterations compared to the aforementioned schemes. Thus, our proposed scheme has a certain delay of reaching the same accuracy of the other schemes for a few number of iterations. However, using the efficient overlapped clustering and aggregated models dissemination between CHs, the proposed FL-EOCD algorithm efficiently disseminates the aggregated models of CHs to their adjacent clusters. Consequently, with increasing the number of iterations, the CHs have better knowledge about the global model of the system. Thus, the accuracy of our proposed FL-EOCD scheme is quickly increased with the number of iterations. Essentially, when the number of global iterations in the range of $150-200$, the CHs in the proposed scheme have almost the same global model of the BS of the star-based and hierarchical FL schemes. Therefore, as depicted from Fig.  \ref{fig4}, all the considered schemes have almost the same number of global iterations upon reaching the accuracy of $96\%$. 

In Fig.  \ref{fig5}, we further study the impact of the number of overlapped clusters on the convergence rate of our proposed FL-EOCD scheme compared to the benchmark schemes for a network of  $22$ RRBs and $26$ devices which are clustered into $4$ overlapped clusters. From Fig. \ref{fig5}, when the number of iterations is $40$, FL-EOCD reaches an accuracy of $88\%$ for $4$ overlapped clusters as opposed to $90\%$ for $3$ overlapped clusters as depicted in Fig. \ref{fig4}. Thus, the performance gain of FL-EOCD is not degraded
for a few number of iterations with relatively large number of clusters, thanks to the efficient overlapped clustering and models' dissemination achieved by FL-EOCD. Again, for a large number of global iterations, FL-EOCD achieves the same accuracy as the FL benchmark schemes without the need for global aggregations at the BS.

\begin{figure}[t!]
		\centerline{\includegraphics[width=0.85\linewidth]{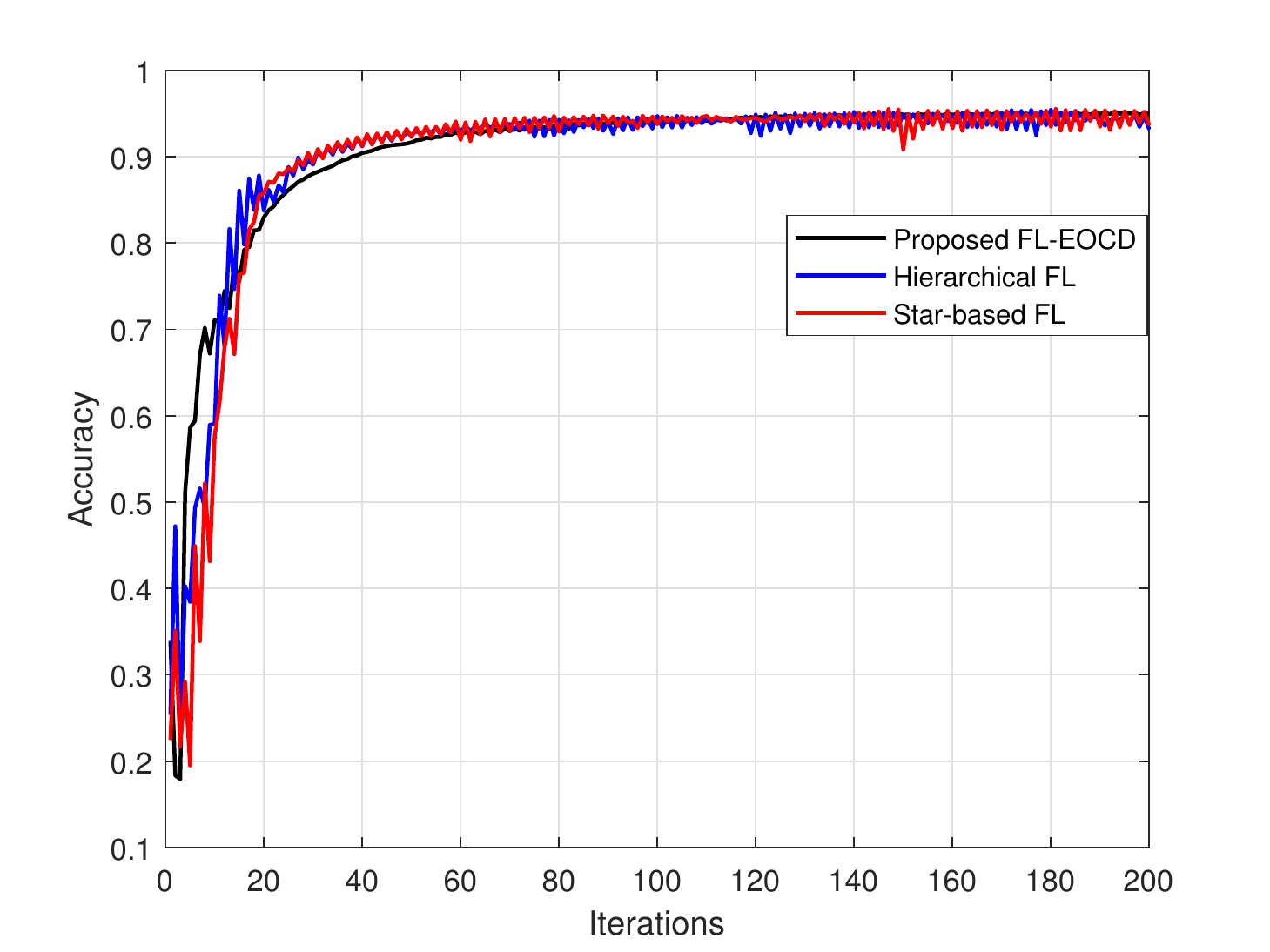}}
		
		\caption{Accuracy vs. number of global iterations for a network of $26$ devices and $22$ RRBs for the configuration of $3$ overlapped clusters.}
		\label{fig4}
\end{figure}

\begin{figure}[t!]
	\centerline{\includegraphics[width=0.85\linewidth]{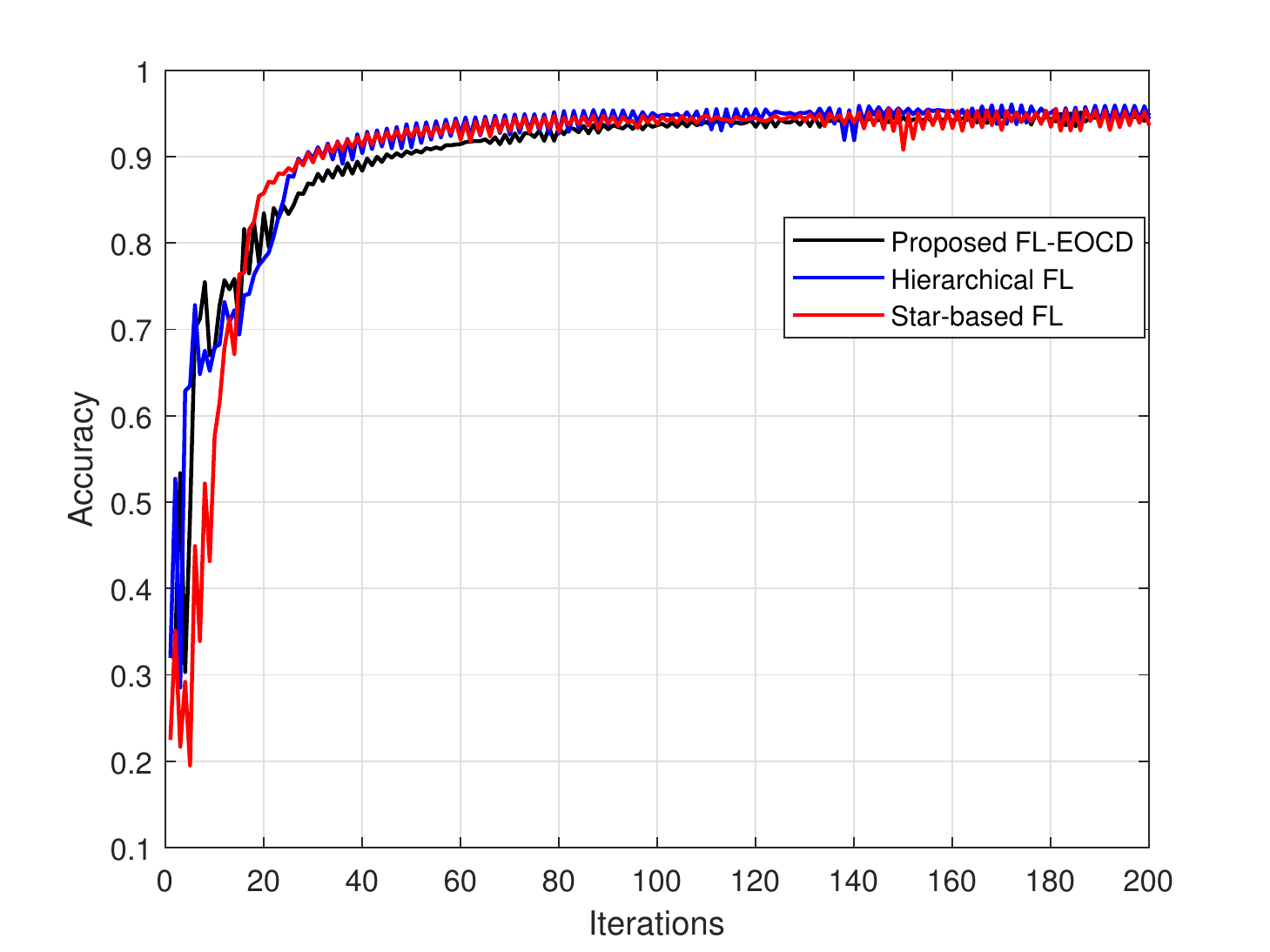}}
	
	\caption{Accuracy vs. number of global iterations for a network of $26$ devices and $22$ RRBs for the configuration of $4$ overlapped clusters.}
	\label{fig5}
\end{figure}

\begin{figure}[t!]
	\centerline{\includegraphics[width=0.85\linewidth]{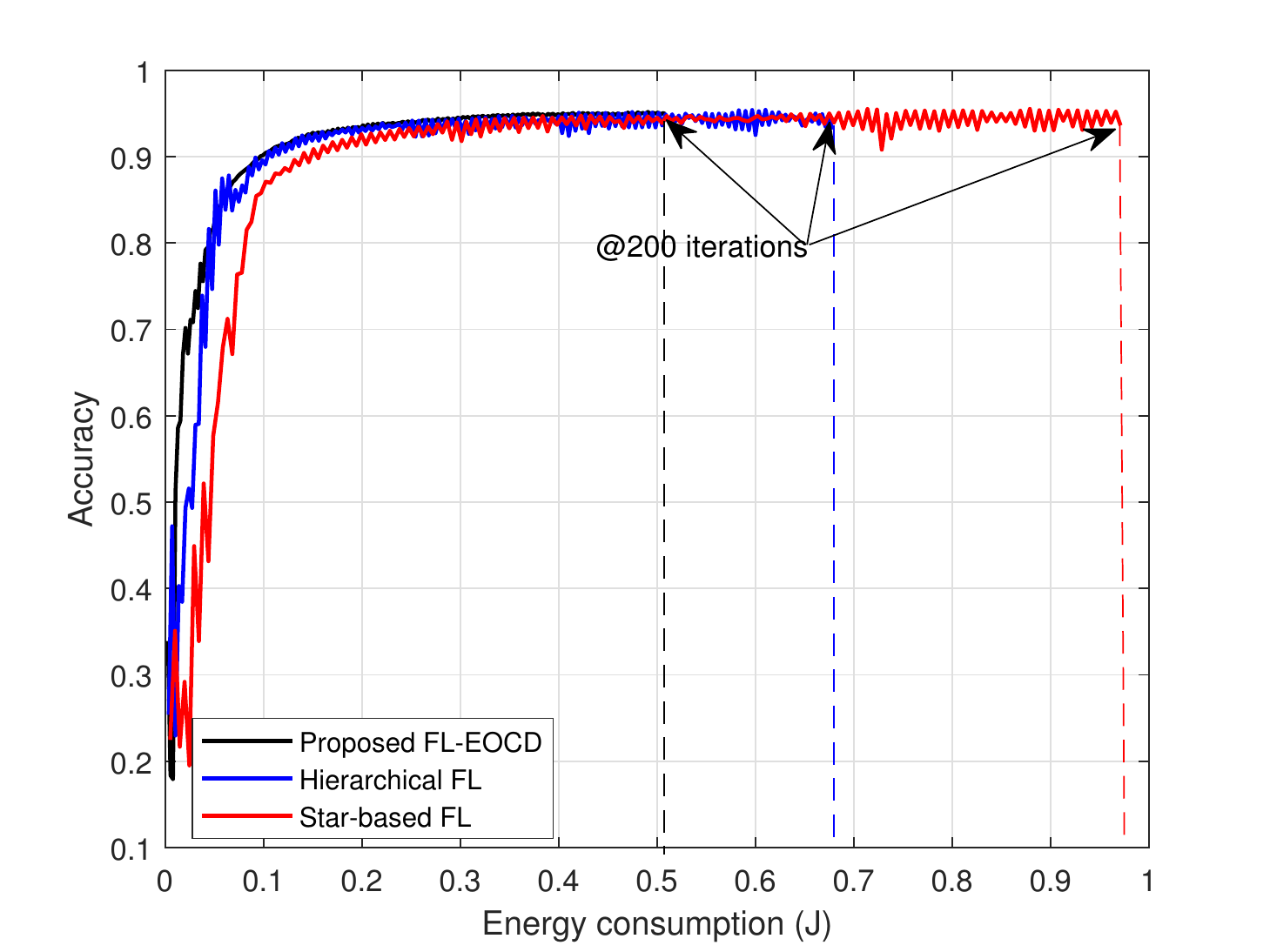}}
	
	\caption{Accuracy vs. energy consumption for a network of $26$ devices and $22$ RRBs.}
	\label{fig6}
\end{figure}

\begin{figure}[t!]
	\centerline{\includegraphics[width=0.85\linewidth]{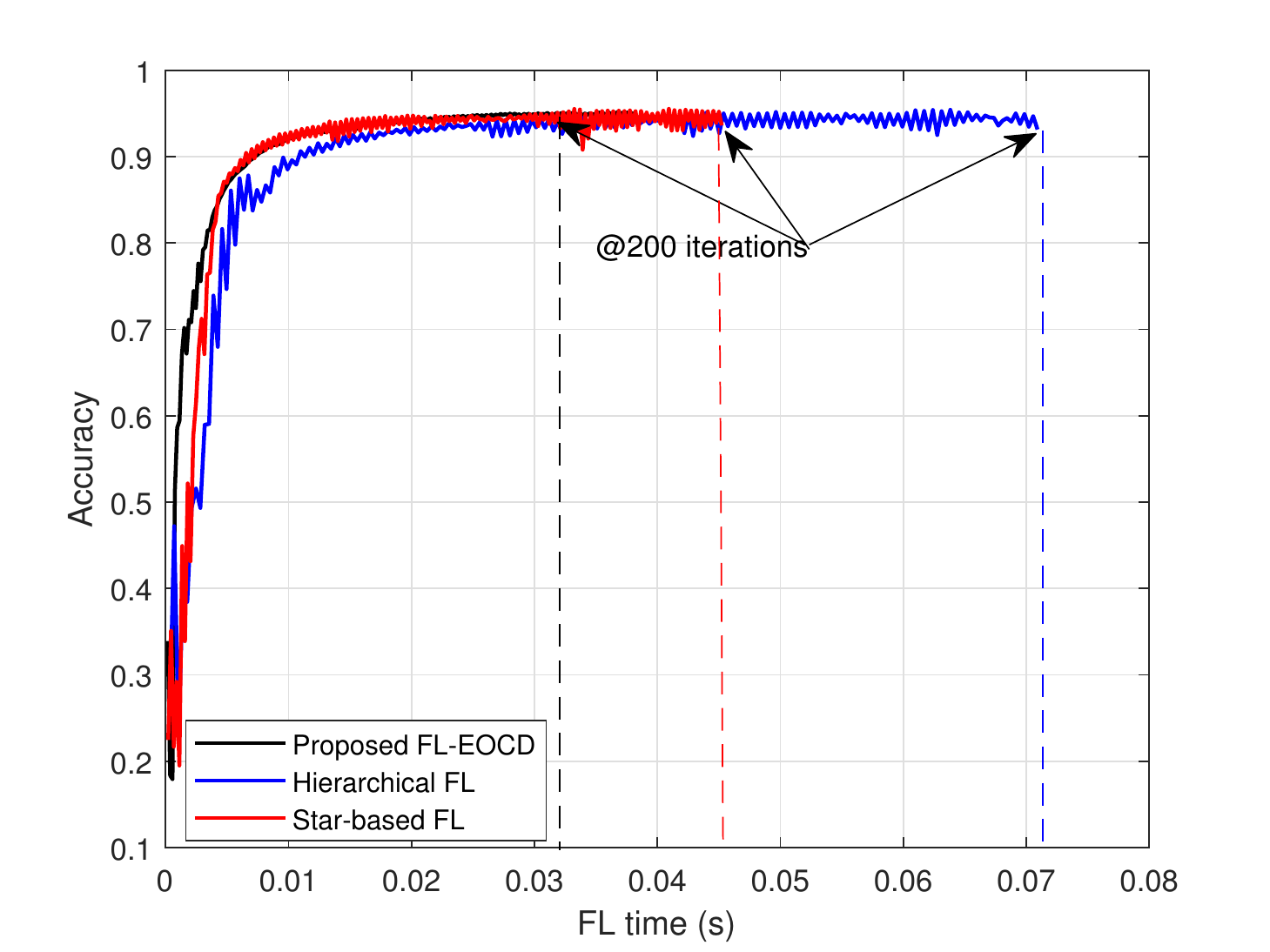}}
	
	\caption{Accuracy vs. FL time for a network of $26$ devices and $22$ RRBs.}
	\label{fig7}
\end{figure}

%\vspace{-0.3cm}
\subsection{Performance of the Proposed FL-EOCD and Benchmark Schemes in terms of Energy Consumption and FL Time}
%\vspace{-0.25cm}
In Fig. \ref{fig6}, we evaluate the accuracy of the proposed FL-EOCD and benchmark schemes versus the required consumed energy in a network of $26$ devices and $22$ RRBs. From Fig. \ref{fig4}, we observe that for same accuracy and number of global iterations, the energy consumption of the proposed FL-EOCD is reduced. Such
an observation can be explained by the following argument. In the proposed FL-EOCD scheme, we judiciously select the CHs such that they have good reachability to many close-by local devices. Thus, the consumed energy for wireless transmissions from devices/BDs to CHs and then from CHs back to devices/BDs is low. However, in the considered star-based FL, devices consume most of the energy for transmitting their trained local models to the distant BS. Consequently, a degradation of the energy consumption is observed at $200$ iterations when the star-based FL converges. It is recalled that in the hierarchical FL,
unlike the proposed FL-EOCD,  two-hop  transmission is required which consumes more energy. As a result, FL-EOCD achieves an improved
energy consumption compared to the considered hierarchical FL scheme. Fig. \ref{fig4} demonstrates
that FL-EOCD requires energy of $0.51$ J compared to the consumed energy of the considered star-based FL and hierarchical schemes of $0.98$ J and $0.69$ J, respectively, for $200$ global iterations. Furthermore, since FL-EOCD requires polynomial computational complexity and has an improved performance in terms of energy consumption, it is advantageous compared to the benchmark FL schemes.

In Fig. \ref{fig7}, we evaluate the accuracy of the proposed FL-EOCD and benchmark schemes versus the FL time in a network of $26$ devices and $22$ RRBs. First, it is clear that the FL time depends on the wireless transmission time and the computation learning time of devices. Since all the simulated schemes consider local learning at the devices and optimize the frequency allocation computation, their local learning time is almost the same. Thus, the key factor in the FL time is  the communication time of transmitting the trained local models to the CHs/BS plus the transmitting time of sending the aggregated models from CHs/BS back to the devices. Consequently, the FL time is dominated by the wireless transmission time. As can be seen from Fig. \ref{fig5}, our proposed FL-ECD scheme, that implements simultaneous and short D2D communications, effectively minimizes the FL time with the same accuracy of the considered benchmark schemes. In particular, it can be observed that the FL time of FL-EOCD at $200$ iterations is $0.031$ second while the considered star-based and hierarchical schemes have FL times of $0.045$ second  and $0.071$ second at the same number of iterations, respectively. This is because the FL time is mainly controlled by the longest transmission time of one device, named as a straggler device, which comes in favor of our proposed FL-EOCD scheme, thanks to the clustering optimization, CH-device scheduling, computation frequency allocation, and D2D communications.

Figs. \ref{fig6} and \ref{fig7} interestingly depict that the star-based FL has a poor performance in terms of energy consumption and improved FL time compared to the  hierarchical FL. This argument is due to the fact that hierarchical FL has good and short connectivity among devices which consumes less energy, however it needs two-hop transmission that degrades the FL latency. Our proposed FL-EOCD strikes a balance between these benchmark schemes by (i) augmenting the need for two hops of transmission (unlike hierarchical FL); and (ii) ensuring short-range of communications between CHs and devices/BDs (unlike the star-based FL that requires a distant BS).

\begin{figure}[t!]
	\centerline{\includegraphics[width=0.85\linewidth]{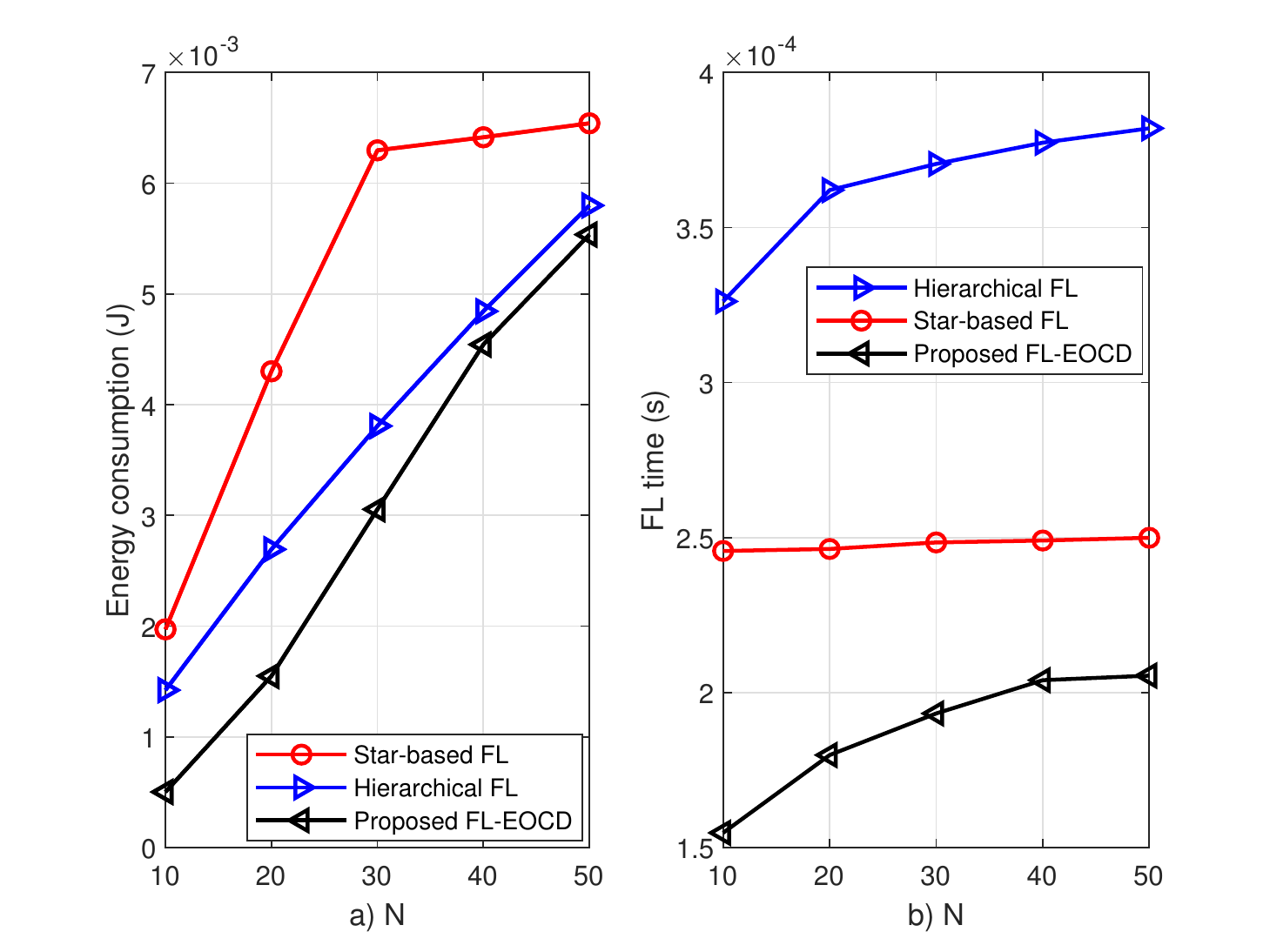}}
	
	\caption{a) Energy consumption and b) FL time per one global iteration vs. number of devices $N$ for a network of $30$ RRBs}
	\label{fig8}
\end{figure}

\begin{figure}[t!]
	\centerline{\includegraphics[width=0.85\linewidth]{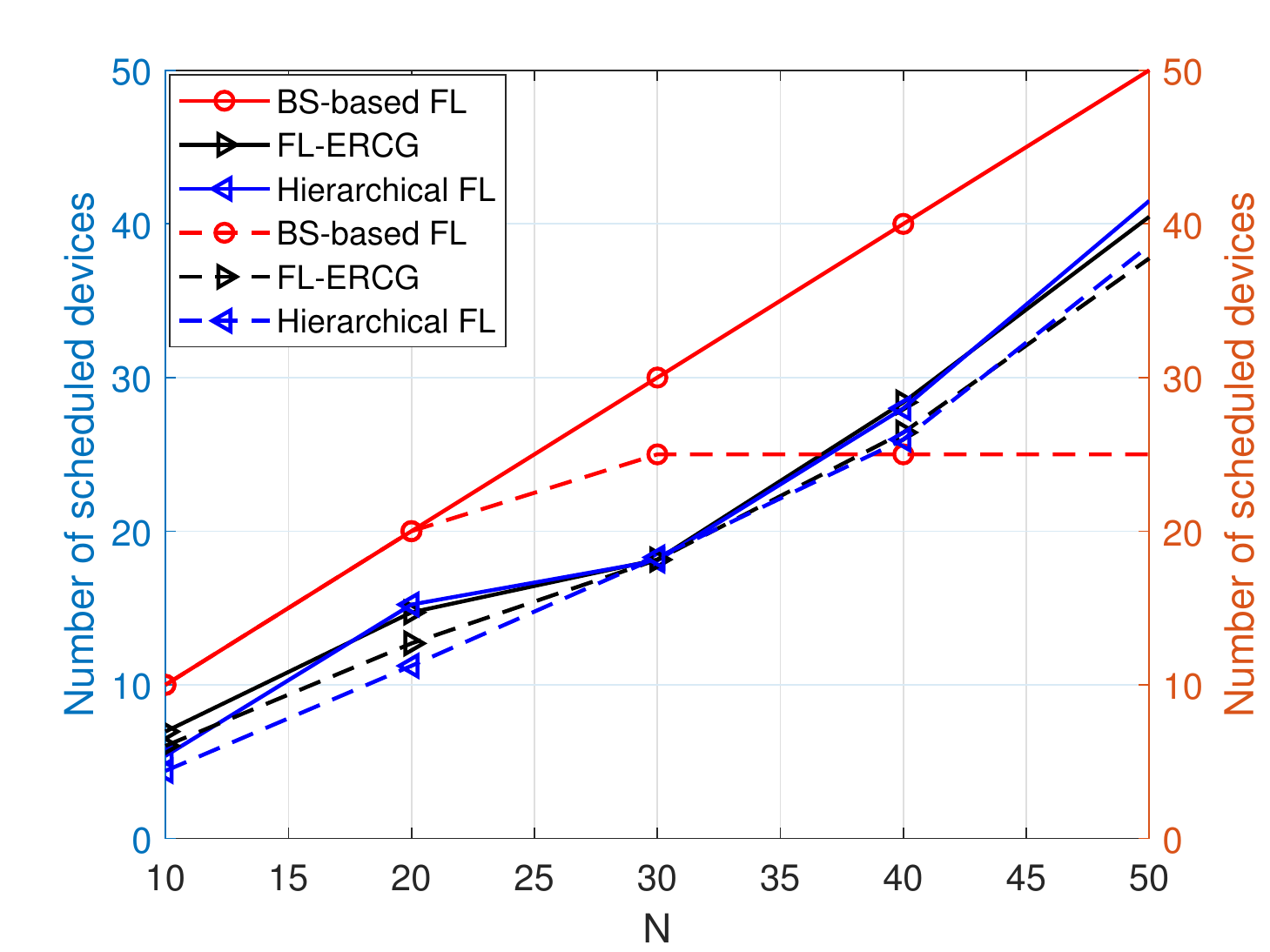}}
	
	\caption{Number of scheduled devices per one global iteration
		vs. number of devices $N$ for a network of $50$ RRBs (solid lines) and $25$ RRBs (dash lines).}
	\label{fig9}
\end{figure}

\subsection{Performance of the Proposed FL-EOCD and Benchmark Schemes for Different Number of Devices and RRBs}
In Figs. \ref{fig8}(a) and  \ref{fig8}(b), we evaluate the energy consumption and FL time for one global iteration, respectively, of the proposed FL-EOCD and benchmark schemes by varying the number of devices in a network of $30$ RRBs. From Fig. \ref{fig8}, we observe that FL-EOCD outperforms the benchmark schemes in terms of energy consumption and FL time. For instance, Fig. \ref{fig8}(a) shows that FL-EOCD achieves
$0.0027$ J and $0.0013$ J lower energy consumption, respectively, than
the considered star-based and hierarchical schemes for $20$ devices. It is noteworthy that when $N$ increases, the consumed energy of all the schemes increases. However, due to the D2D short communications and efficient overlapped clustering of our proposed FL-EOCD scheme, the power consumption degradation for increasing $N$ is insignificant, thanks to the joint clustering and device scheduling achieved  by our proposed scheme for the increased number of devices. Indeed, Fig. \ref{fig8}(a) shows that FL-EOCD achieves $0.0032$ J  lower energy consumption compared the considered star-based scheme for $30$ devices. Similarly, Fig. \ref{fig8}(b) shows that FL-EOCD achieves $0.17$ ms and $0.05$ ms shorter FL time than the considered  hierarchical and star-based schemes for $50$ devices, respectively. It is noted that the star-based FL can support at a maximum of $30$ devices that is equal to the number of available RRBs in the system. Therefore, the energy consumption and FL time performances of the star-based FL do not grow up with increasing the number of devices ($N>30$). Based on Fig. \ref{fig8}(a) and Fig. \ref{fig8}(b), we conclude that the proposed FL-EOCD
algorithm is efficient in terms of energy consumption and FL time for both small and large numbers of devices.

%It is noteworthy that when $N$ increases, the consumed energy of all the schemes increases. However, due to the D2D short communications and novel overlapped clustering of our proposed scheme, the power consumption degradation for increasing $N$ is not significant. Thus, the performance gain of FL-ECD is enhanced for a large numbers of devices, thanks to the joint clustering and devices scheduling achieved  by our proposed scheme for the increased number of devices. Indeed, Fig. \ref{fig6}-a shows that FL-ECD achieves $0.0140$ J and $0.001$ J lower energy consumption, respectively, than the considered BS-based and hierarchical schemes for $50$ devices. Similarly, Fig. \ref{fig6}-b shows that FL-ECD achieves $15$ ms and $5$ ms lower FL time
%than the considered schemes for $50$ devices, respectively. Based on Fig. \ref{fig6}-a and Fig. \ref{fig6}-b, we conclude that the proposed FL-ECD
%algorithm is efficient in terms of energy consumption and FL time for both small and large numbers of devices.

In Fig. \ref{fig9}, we evaluate the number of scheduled devices of the proposed FL-EOCD and benchmark schemes by changing the number of devices in a network of $50$ RRBs (solid lines) and $25$ RRBs (dash lines). From Fig. \ref{fig9}, we observe that for both FL-EOCD and hierarchical schemes, the number of scheduled devices
is degraded as compared to the considered star-based FL scheme. Both schemes have a certain degradation, around $7\%$, compared to the star-based scheme where the number of scheduled devices grows relatively fast. This is becasue in the considered star-based FL scheme, the BS can access all the devices, and thus it can schedule all the devices in the network, given the number of RRBs in the network. On the other hand, in FL-EOCD scheme, the selected overlapped CHs have limited coverage range and cannot schedule all the devices in the network. Following this, the remaining devices are unscheduled. Consequently, a degradation
in the number of scheduled devices is observed in case of large numbers of devices in the network. 

Fig. \ref{fig9} also illustrates that the proposed FL-EOCD algorithm and hierarchical scheme outperform the considered star-based scheme when the number of RRBs is less than the number of devices. In particular, when $Z$ is nearly $25$,
the effective system capacity of the star-based scheme stops growing and can have at most $25$
scheduled devices. Therefore, the star-based scheme is impractical in terms of the number of scheduled devices, particularly for massive networks with limited radio resources. However, in the considered D2D schemes (i.e., FL-EOCD and hierarchical), the set of RRBs can be re-used among non-adjacent clusters. As a result, the number of scheduled devices of the considered D2D schemes is increased. Indeed, Fig. \ref{fig9} (dash lines) demonstrates
that D2D schemes achieve more than $25\%$ improvement in the
number of scheduled devices compared to the star-based FL
for $50$ devices and $25$ RRBs in the network.

\section{Concluding Remarks and Future Works} \label{CN}
%\vspace{-0.1cm}
In this paper, we have introduced a new FL-EOCD scheme, that leverages D2D communications and overlapped clustering without the need for a global aggregator. %For this objective, first we have  mathematically and numerically proofed the convergence rate of our developed FL-ERCG scheme. %Our conducted numerical results showed that the proposed FL-ERCG scheme reaches almost the same performance of the BS-based scheme. 
In order to study the performance of the  developed FL-EOCD scheme, second we have investigated the resource allocation strategy to minimize the energy consumption in a partially connected D2D network subject to FL time constraint. Our proposed resource allocation strategy jointly optimizes the clustering selection, CH-device scheduling, and frequency computation allocation. The presented numerical results revealed that the proposed FL-EOCD scheme substantially reduces the energy	consumption and FL time compared to the baseline schemes. In particular, simulation results showed that	the proposed FL-EOCD scheme can effectively				improve the energy consumption by around $35\%$ and $55\%$, and FL time by around $50\%$ and $30\%$, respectively, compared to  the hierarchical and  star-based schemes. For the required accuracy of $80\%$, our proposed FL-EOCD scheme has a certain degradation in the number of global iterations as compared to the benchmark schemes. This small degradation in some numerical results, roughly in the range of
$3$-$5$ global iterations, comes at the achieved: (i) low energy consumption and (ii) reduced FL time without using a global aggregator as compared to the benchmark schemes.

In this paper, due to the space limitation, we have numerically proved the convergence rate of our proposed FL-EOCD scheme, and the mathematical analysis of the convergence rate is left for a future work.  As an extension, it is also possible to alleviate the existence of a global coordinator (e.g., BS) by developing a game theoretical approach.
%\vspace{-0.4cm}

%\end{spacing}

\end{document}